\documentclass{sigplanconf}
\usepackage{amsmath,amssymb,latexsym,amsthm}
\usepackage{upgreek}
\usepackage[
  pdftitle={Abstracting Abstract Machines},
  pdfauthor={David Van Horn and Matthew Might},
  pdfsubject={Computer Science},
  pdfkeywords={abstract machines, abstract interpretation},
  breaklinks=true
]
{hyperref}
\newcommand{\ie}{\emph{i.e.}}
\newcommand{\eg}{\emph{e.g.}}

\newcommand{\etal}{\emph{et al.}}

\newcommand{\syn}[1]{\mathit{#1}} 
\newcommand{\var}[1]{\mathit{#1}}
\newcommand{\s}[1]{\mathit{#1}}

\newcommand{\parto}{\rightarrow_{\mathrm{fin}}}
\newcommand{\dom}{\var{dom}}

\newcommand{\set}[1]{\left\{#1\right\}}
\newcommand{\setbuild}[2]{\left\{ #1 : #2\right\}}

\newcommand{\Pow}[1]{{\mathcal{P}\left(#1\right)}}

\newcommand{\union}{\cup}

\newcommand{\wt}{\sqsubseteq}

\newcommand{\join}{\sqcup}

\newcommand{\bigjoin}{\bigsqcup}

\DeclareMathOperator{\lfp}{lfp}

\newenvironment{grammar}{\begin{array}{r@{\;}c@{\;}l@{\;}c@{\;}l@{\;\;\;\;}l}}{\end{array}}

\newcommand{\opor}{\mathrel{|}}

\newcommand{\produces}{\mathrel{::=}}

\newcommand{\vv}{x}

\newcommand{\schfalse}{{\mbox{\tt\#f}}}

\newcommand{\ttlp}{\mbox{\tt (}}
\newcommand{\ttrp}{\mbox{\tt )}}

\newcommand{\appform}[2]{\ttlp #1 #2\ttrp}
\newcommand{\lamform}[2]{\ttlp \uplambda #1.#2\ttrp}

\newcommand{\ifform}[3]{\ttlp {\tt if}\; #1\; #2\; #3\ttrp}

\newcommand{\throwform}[1]{\ttlp {\tt throw}\; #1\ttrp}

\newcommand{\lab}{\ell}

\newcommand{\expr}{e}

\newcommand{\State}{\Sigma}
\newcommand{\state}{\varsigma}

\newcommand{\Env}{\s{Env}}

\newcommand{\Den}{\s{Val}} 

\newcommand{\store}{\sigma}

\newcommand{\env}{\rho}

\newcommand{\cont}{\kappa}

\newcommand{\alloc}{\mathit{alloc}}

\newcommand{\addr}{a}

\newcommand{\tm}{t}
\newcommand{\tick}{{{tick}}}

\newcommand{\sa}[1]{\widehat{\mathit{#1}}}

\newcommand{\aState}{{\hat{\Sigma}}}
\newcommand{\astate}{{\hat{\varsigma}}}

\newcommand{\astore}{{\hat{\sigma}}}

\newcommand{\aaddr}{{\hat{\addr}}}
\newcommand{\aalloc}{{\widehat{alloc}}}

\newcommand{\atick}{{\widehat{tick}}}

\newcommand{\absmap}{\alpha}
\newcommand{\abs}[1]{|#1|}

\newcommand\CESP{CESK$^\star$}

\newcommand{\multistep}{\longmapsto\!\!\!\!\!\rightarrow}

\newcommand{\handls}{\eta}
\newcommand{\handk}{\mathbf{hn}}

\newcommand{\catchform}[2]{\ttlp {\tt catch}\; #1\; #2\ttrp}

\newcommand\tmnext{u}
\newcommand\haddr{h}
\newcommand\addrnext{b}
\newcommand\addrnextnext{c}
\newcommand\Storable{Storable}
\newcommand\sto{s}
\newcommand\den{v}

\newcommand\Perm{\mathcal{P}}

\newcommand\mtk{\mathbf{mt}}
\newcommand\fnk{\mathbf{fn}}
\newcommand\ark{\mathbf{ar}}
\newcommand\kok{\mathbf{c_1}}
\newcommand\ktk{\mathbf{c_2}}
\newcommand\computed{\mathbf{c}}
\newcommand\delayed{\mathbf{d}}

\newcommand\mte{\emptyset} 
\newcommand\mts{\emptyset} 

\newcommand\liveloc{\mathcal{LL}}

\newcommand\no{\text{deny}}
\newcommand\grant{\text{grant}}

\newcommand\OK{\mathcal{OK}}
\newcommand\AOK{\widehat{\mathcal{OK}^\star}}

\newcommand\betavalue{{\beta_{\mathsf{v}}}}

\newcommand\grantform[2]{\ttlp{\tt grant}\; #1\; #2\ttrp}
\newcommand\testform[3]{\ttlp{\tt test}\; #1\; #2\; #3\ttrp}
\newcommand\failexpr{{\tt fail}}
\newcommand\frameform[2]{\ttlp{\tt frame}\; #1\; #2\ttrp}

\newcommand\evf{\mathit{eval}} 
\newcommand\avf{\mathit{aval}} 
\newcommand\inj{\mathit{inj}} 
\newcommand\fv{\mathbf{fv}}

\newcommand\restrict[2]{#1 | #2}

\newtheorem{theorem}{Theorem}
\newtheorem{lemma}{Lemma}

\begin{document}

\conferenceinfo{ICFP'10,} {September 27--29, 2010, Baltimore, Maryland, USA.}
\CopyrightYear{2010}
\copyrightdata{978-1-60558-794-3/10/09}

\title{Abstracting Abstract Machines}

\authorinfo{David {Van Horn}\thanks{Supported by the National Science
    Foundation under grant 0937060 to the Computing Research
    Association for the CIFellow Project.}}
           {Northeastern University}
           {dvanhorn@ccs.neu.edu}

\authorinfo{Matthew Might}
           {University of Utah}
           {might@cs.utah.edu}
\maketitle

\begin{abstract}
  We describe a derivational approach to abstract interpretation that
  yields novel and transparently sound static analyses when applied to
  well-established abstract machines.  To demonstrate the technique
  and support our claim, we transform the CEK machine of Felleisen and
  Friedman, a lazy variant of Krivine's machine, and the
  stack-inspecting CM machine of Clements and Felleisen into abstract
  interpretations of themselves.  The resulting analyses bound
  temporal ordering of program events; predict return-flow and
  stack-inspection behavior; and approximate the flow and evaluation
  of by-need parameters.  For all of these machines, we find that a
  series of well-known concrete machine refactorings, plus a technique
  we call store-allocated continuations, leads to machines that
  abstract into static analyses simply by bounding their stores.  We
  demonstrate that the technique scales up uniformly to allow static
  analysis of realistic language features, including tail calls,
  conditionals, side effects, exceptions, first-class continuations,
  and even garbage collection.
\end{abstract}




\section{Introduction}

Abstract machines such as the CEK machine and Krivine's machine are
first-order state transition systems that represent the core of a real
language implementation.
Semantics-based program analysis, on the other hand, is concerned with
safely approximating intensional properties of such a machine as it
runs a program.
It seems natural then to want to systematically derive analyses from
machines to approximate the core of realistic run-time systems.

Our goal is to develop a technique that enables direct abstract
interpretations of abstract machines by methods for transforming a
given machine description into another that computes its finite
approximation.

%
%
 
We demonstrate that the technique of refactoring a machine with
\textbf{store-allocated continuations} allows a direct structural abstraction\footnote{
  A structural abstraction 
  distributes 
  component-, 
  point-, 
  and
  member-wise. 
} by
bounding the machine's store.
Thus, we are able to convert semantic techniques used to model
language features into static analysis techniques for reasoning about
the behavior of those very same features.
By abstracting well-known machines, our technique delivers static
analyzers that can reason about by-need evaluation, higher-order
functions, tail calls, side effects, stack structure, exceptions and
first-class continuations.

The basic idea behind store-allocated continuations is not new.
SML/NJ has allocated continuations in the heap for well over a
decade~\cite{dvanhorn:Shao1994Spaceefficient}.
At first glance, modeling the program stack in an abstract machine
with store-allocated continuations would not seem to provide any real
benefit.
Indeed, for the purpose of defining the meaning of a program, there is
no benefit, because the meaning of the program does not depend on the
stack-implementation strategy.
Yet, a closer inspection finds that store-allocating continuations
eliminate recursion from the definition of the state-space of the
machine.
With no recursive structure in the state-space, an abstract machine
becomes eligible for conversion into an abstract interpreter through
a simple structural abstraction.


To demonstrate the applicability of the approach, we derive abstract
interpreters of:
\begin{itemize}
\item a call-by-value $\lambda$-calculus with state and control based
  on the CESK machine of Felleisen and
  Friedman~\cite{dvanhorn:Felleisen1987Calculus},

\item a call-by-need $\lambda$-calculus based on a tail-recursive,
  lazy variant of Krivine's machine derived by Ager, Danvy and
  Midtgaard~\cite{dvanhorn:Ager2004Functional}, and

\item a call-by-value $\lambda$-calculus with stack inspection based
  on the CM machine of Clements and
  Felleisen~\cite{dvanhorn:Clements2004Tailrecursive};
\end{itemize}
and use abstract  garbage
collection to improve 
precision~\cite{mattmight:Might:2006:GammaCFA}.

\subsection*{Overview}

In Section~\ref{sec:cek-to-acesk}, we begin with the CEK machine and
attempt a structural abstract interpretation, but find ourselves
blocked by two recursive structures in the machine: environments and
continuations.
We make three refactorings to:
\begin{enumerate}
\item store-allocate bindings,
\item store-allocate continuations, and 
\item time-stamp machine states;
\end{enumerate}
resulting in the CESK, CESK$^\star$, and time-stamped CESK$^\star$
machines, respectively.
The time-stamps encode the history (context) of the machine's
execution and facilitate context-sensitive abstractions.
We then demonstrate that the time-stamped machine abstracts directly
into a parameterized, sound and computable static analysis.

In Section~\ref{sec:krivine}, we replay this process (slightly
abbreviated) with a lazy variant of Krivine's machine to arrive at a
static analysis of by-need programs.

In Section~\ref{sec:realistic-features}, we incorporate
conditionals, side effects, exceptions,
first-class continuations, and garbage collection.

In Section~\ref{sec:CM}, we abstract the CM (continuation-marks)
machine to produce an abstract interpretation of stack inspection.

In Section~\ref{sec:widening}, we widen the abstract interpretations
with a single-threaded ``global'' store to accelerate convergence.
For some of our analyzers, this widening results in polynomial-time
algorithms and connects them back to known analyses.

\section{From CEK to the abstract CESK\texorpdfstring{$^\star$}{*}}
\label{sec:cek-to-acesk}

In this section, we start with a traditional machine for a programming
language based on the call-by-value $\lambda$-calculus, and gradually
derive an
abstract interpretation of this machine.
%
The outline followed in this section covers the basic steps for
systematically deriving abstract interpreters that we follow
throughout the rest of the paper.

To begin, consider the following language of expressions:%
\footnote{\emph{Fine print on syntax:}
As is often the case in program analysis where \emph{semantic values}
are approximated using \emph{syntactic phrases} of the program under
analysis, we would like to be able to distinguish different syntactic
occurrences of otherwise identical expressions within a
program.  Informally, this means we want
to track the source location of expressions.  Formally, this is
achieved by \emph{labeling} expressions and assuming all labels within
a program are distinct:
\[
\begin{grammar}
  \expr &\in& \syn{Exp} &\produces& 
  \vv^\lab \opor \appform{\expr}{\expr}^\lab \opor \lamform{\vv}{\expr}^\lab
\\
\lab &\in& \syn{Lab} & & \text{an infinite set of labels}
\text.
\end{grammar}
\]
However, we judiciously omit labels whenever they are irrelevant and
doing so improves the clarity of the presentation. 
Consequently, they appear only in Sections~\ref{sec:labelled-acesk}
and~\ref{sec:widening}, which are concerned with $k$-CFA.
}
\[
\begin{grammar}
  \expr &\in& \syn{Exp} &\produces& 
 \vv \opor \appform{\expr}{\expr} \opor   \lamform{\vv}{\expr}
\\
  \vv &\in& \syn{Var} & &\text{ a set of identifiers}
  \text.
\end{grammar}
\]

A standard machine for evaluating this language is the CEK machine of
Felleisen and Friedman~\cite{mattmight:Felleisen:1986:CEK}, and it is
from this machine we derive the abstract semantics---a computable
approximation of the machine's behavior.
Most of the steps in this derivation correspond to
well-known machine transformations and real-world implementation
techniques---and most of these steps are concerned only with the
\emph{concrete machine}; a very simple abstraction is employed only at
the very end.

The remainder of this section is outlined as follows: we present the
CEK machine, to which we add a store, and use it to allocate variable
bindings.  This machine is just the CESK machine of Felleisen and
Friedman~\cite{dvanhorn:Felleisen1987Calculus}.  From here, we further
exploit the store to allocate continuations, which corresponds to a
well-known implementation technique used in functional language
compilers~\cite{dvanhorn:Shao1994Spaceefficient}.  We then abstract
\emph{only the store} to obtain a framework for the sound, computable
analysis of programs.

\subsection{The CEK machine}
\label{sec:cek}

A standard approach to evaluating programs is to rely on a
Curry-Feys-style Standardization Theorem, which says roughly: if an
expression $\expr$ reduces to $\expr'$ in, \eg, the call-by-value
$\lambda$-calculus, then $\expr$ reduces to $\expr'$ in a canonical
manner.  This canonical manner thus determines a state machine for
evaluating programs: a standard reduction machine.

To define such a machine for our language, we define a grammar of
evaluation contexts and notions of reduction (\eg, $\betavalue$).  An
evaluation context is an expression with a ``hole'' in it.  For
left-to-right evaluation order, we define evaluation contexts $E$ as:
\[
\begin{grammar}
 && E &\produces& [\;] \opor \appform{E}{\expr} \opor \appform{\den}{E}\text.
\end{grammar}
\]
An expression is either a value or uniquely decomposable into an
evaluation context and redex.  The standard reduction machine is:
\[
E[\expr] \longmapsto_\betavalue E[\expr'],\mbox{ if } \expr\;\betavalue\;\expr'\text.
\]
However, this machine does not shed much light on a realistic
implementation.  At each step, the machine traverses the entire
source of the program looking for a redex.  When found, the redex is
reduced and the contractum is plugged back in the hole, then the
process is repeated.

Abstract machines such as the CEK machine, which are derivable from
standard reduction machines, offer an extensionally equivalent but
more realistic model of evaluation that is amenable to efficient
implementation.  The CEK is environment-based; it uses environments
and closures to model substitution.  It represents evaluation contexts
as \emph{continuations}, an inductive data structure that models
contexts in an inside-out manner.  The key idea of machines such as
the CEK is that the whole program need not be traversed to find the
next redex, consequently the machine integrates the process of
plugging a contractum into a context and finding the next redex.

States of the CEK machine~\cite{mattmight:Felleisen:1986:CEK}
consist of a control string (an expression), an environment that
closes the control string, and a continuation:
\[
\begin{grammar}
 \state &\in& \State &=& 
  \syn{Exp} \times \s{Env} \times \s{Kont} 
\\
   \den &\in& \s{Val} &\produces& \lamform{\vv}{\expr}
\\
  \env &\in &\s{Env} &=& \syn{Var} \parto \s{Val} \times \s{Env}
\\
  \kappa &\in& \s{Kont} &\produces&
  \mtk \opor \ark(\expr,\env,\kappa) \opor \fnk(\den,\env,\kappa)\text.
\end{grammar}
\]
States are identified up to consistent renaming of bound variables.

Environments are finite maps from variables to closures.
Environment extension is written $\env[\vv\mapsto(\den,\env')]$.

Evaluation contexts $E$ are represented (inside-out) by continuations
as follows: $[\; ]$ is represented by $\mtk$; $E[\appform{[\;
  ]}{\expr}]$ is represented by $\ark(\expr',\env,\cont)$ where $\env$
closes $\expr'$ to represent $\expr$ and $\cont$ represents $E$;
$E[\appform{\den}{[\; ]}]$ is represented by $\fnk(\den',\env,\cont)$
where $\env$ closes $\den'$ to represent $\den$ and $\cont$ represents $E$.

\begin{figure}
\[
\begin{array}{@{}l@{\qquad}|@{\qquad}r@{}}
\multicolumn{2}{c}{\state \longmapsto_{\mathit{CEK}} \state'} \\[1mm]
\hline 
\\
\langle\vv, \env, \cont\rangle &
\langle\den,\env', \cont\rangle
 \mbox{ where }\env(\vv) = (\den,\env') 
\\[1mm]
\langle\appform{\expr_0}{\expr_1}, \env, \cont\rangle &
\langle\expr_0, \env, \ark(e_1, \env, \cont)\rangle
\\[1mm]
\langle\den,\env, \ark(e,\env',\cont)\rangle &
\langle\expr,\env',\fnk(\den,\env,\cont)\rangle
\\[1mm]
\langle\den,\env, \fnk(\lamform{\vv}{\expr},\env',\cont)\rangle &
\langle\expr,\env'[\vv \mapsto (\den,\env)], \cont\rangle
\end{array}
\]
\caption{The CEK machine.}
\label{fig:cek}
\end{figure}

The transition function for the CEK machine is defined in
Figure~\ref{fig:cek} (we follow the textbook treatment of the CEK
machine~\cite[page 102]{dvanhorn:Felleisen2009Semantics}).
The initial machine state for a closed expression $\expr$
is given by the $\inj$ function:
\[
\inj_{\mathit{CEK}}(\expr) = \langle\expr,\mte,\mtk\rangle\text.
\]
Typically, an evaluation function is defined as a partial function from
closed expressions to answers:
\[
\evf'_{\mathit{CEK}}(e) =
 (\den,\env)
 \mbox{ if }
 \inj(\expr) \multistep_{\mathit{CEK}} 
 \langle\den,\env,\mtk\rangle\text.
\]
This gives an extensional view of the machine, which is useful,
\eg, to prove correctness with respect to a canonical
evaluation function such as one defined by standard reduction or
compositional valuation.
However for the purposes of program analysis, we are concerned more
with the intensional aspects of the machine.  As such, 
we define the meaning of a program as the (possibly
infinite) set of reachable machine states:
\[
\evf_{\mathit{CEK}}(\expr) = \{ \state\ |\  \inj(\expr) \multistep_{\mathit{CEK}} \state \}
\text.
\]

Deciding membership in the set of reachable machine states is not
possible due to the halting problem.  The goal of abstract
interpretation, then, is to construct a function,
$\avf_{\widehat{\mathit{CEK}}}$, that is a sound and computable
approximation to the $\evf_{\mathit{CEK}}$ function.

We can do this by constructing a machine that is similar in structure to
the CEK machine:
it is defined by an \emph{abstract state transition} relation
$(\longmapsto_{\widehat{\mathit{CEK}}}) \subseteq
\aState \times
\aState$, which operates over \emph{abstract states}, $\hat{\State}$,
which approximate the states of the CEK machine,
and an abstraction map $\absmap : \State \to \aState$ that maps
concrete machine states into abstract machine states.

The abstract evaluation function is then defined as:
\[
\avf_{\widehat{\mathit{CEK}}}(\expr) = \{ \astate\ |\ \alpha(\inj(\expr))
\multistep_{\widehat{\mathit{CEK}}} \astate \}
\text.
\]

\begin{enumerate}
\item We achieve \emph{decidability} by constructing the approximation in
such a way that the state-space of the abstracted machine is
finite, which guarantees that for any closed expression $\expr$, the
set $\avf(\expr)$ is finite.

\item We achieve \emph{soundness} by demonstrating the abstracted machine
transitions preserve the abstraction map, so that if
 \(\state \longmapsto \state' \text{ and }  \alpha(\state)
  \wt \astate\),
then there exists an abstract state $\astate'$ such that
  \(\astate
\longmapsto  \astate' \text{ and } \alpha(\state') \wt \astate'\).
\end{enumerate}

\paragraph{A first attempt at abstract interpretation:}
%
A simple approach to abstracting the machine's state space is to apply
a {\em structural abstract interpretation}, which lifts abstraction
point-wise, element-wise, component-wise and member-wise across the
structure of a machine state (\ie,~expressions, environments,
and continuations).

The problem with the structural abstraction approach for the CEK
machine is that both environments and continuations are recursive
structures.
As a result, the map $\absmap$ yields objects in an abstract
state-space with recursive structure, implying the space is infinite.
It is possible to perform abstract interpretation over 
an infinite state-space, but it requires a widening operator.
A widening operator accelerates the ascent up the lattice of
approximation and must guarantee convergence.
It is difficult to imagine a widening operator, other than the one
that jumps immediately to the top of the lattice, for these semantics.

Focusing on recursive structure as the source of the problem, a
reasonable course of action is to add a level of indirection to the
recursion---to force recursive structure to pass through explicitly
allocated addresses.
In doing so, we will unhinge recursion in a program's data structures
and its control-flow from recursive structure in the state-space.

We turn our attention next to the CESK
machine~\cite{mattmight:Felleisen:1987:Dissertation,dvanhorn:Felleisen1987Calculus},
since the CESK machine eliminates recursion from one of the structures
in the CEK machine: environments.
In the subsequent section (Section~\ref{sec:ceskp}), we will develop a CESK machine with a pointer
refinement (\CESP{}) that eliminates the other source of recursive structure:
continuations.
At that point, the machine structurally abstracts via a single point
of approximation: the store.

\subsection{The CESK machine}
\label{sec:cesk}

The states of the CESK machine extend those of the CEK machine to
include a \emph{store}, which provides a level of indirection for
variable bindings to pass through.
The store is a finite map from \emph{addresses} to \emph{storable
  values} and 
environments are changed to map variables to addresses.
When a variable's value is looked-up by the machine, it is now
accomplished by using the environment to look up the variable's
address, which is then used to look up the value.
To bind a variable to a value, a fresh location in the store is
allocated and mapped to the value; the
environment is extended to map the variable to that address.

The state space for the CESK machine is defined as follows:
\[
\begin{grammar}
  \state &\in& \State &=& 
  \syn{Exp} \times \s{Env} \times \s{Store} \times \s{Kont}
  \\
  \env &\in &\s{Env} &=& \syn{Var} \parto \s{Addr}
  \\
  \store &\in &\s{Store} &=& \s{Addr} \parto \Storable
  \\
  \sto &\in &\Storable &=& \Den \times \Env
  \\
  \addr,\addrnext,\addrnextnext &\in &\s{Addr} & & \text{ an infinite set}
  \text.
\end{grammar}
\]
States are identified up to consistent renaming of bound variables and
addresses.
The transition function for the CESK machine is defined in
Figure~\ref{fig:cesk} (we follow the textbook treatment of the CESK
machine~\cite[page 166]{dvanhorn:Felleisen2009Semantics}).
\begin{figure}
\[
\begin{array}{@{}l@{\;}|@{\;}r@{}}

\multicolumn{2}{c}{\state \longmapsto_{\mathit{CESK}} \state'} \\[1mm]
\hline
\\
\langle\vv, \env, \store, \cont\rangle &
\langle\den,\env', \store, \cont\rangle \mbox{ where }\store(\env(\vv)) = (\den,\env')
\\[1mm]
\langle\appform{\expr_0}{\expr_1}, \env, \store, \cont\rangle &
\langle\expr_0, \env, \store, \ark(e_1, \env, \cont)\rangle
\\[1mm]
\langle\den,\env,\store, \ark(e,\env',\cont)\rangle &
\langle\expr,\env',\store,\fnk(\den,\env,\cont)\rangle
\\[1mm]
\langle\den,\env,\store, \fnk(\lamform{\vv}{\expr},\env',\cont)\rangle 
&
\langle\expr,\env'[\vv \mapsto \addr], \store[\addr \mapsto (\den,\rho)], \cont\rangle
\\
&
\mbox{ where } \addr\notin\dom(\store)
\end{array}
\]
\caption{The CESK machine.}
\label{fig:cesk}
\end{figure}

The initial state for a closed expression is given by the $\inj$
function, which combines the expression with the empty environment,
store, and continuation:
\[
\inj_{\mathit{CESK}}(\expr) = \langle\expr,\mte,\mts,\mtk\rangle\text.
\]
The $\evf_{\mathit{CESK}}$ evaluation function is defined
following the template of the CEK evaluation given in
Section~\ref{sec:cek}:
\[
\evf_{\mathit{CESK}}(\expr) = \{ \state\ |\  \inj(\expr) \multistep_{\mathit{CESK}} \state \}
\text.
\]
Observe that for any closed expression, the CEK and CESK machines
operate in lock-step: each machine transitions, by the corresponding
rule, if and only if the other machine transitions.
\begin{lemma}[Felleisen, \cite{mattmight:Felleisen:1987:Dissertation}]
\label{lem:store-equiv}
$\evf_{\mathit{CESK}}(e) \simeq \evf_{\mathit{CEK}}(e)$.
\end{lemma}

\paragraph{A second attempt at abstract interpretation:}
With the CESK machine, half the problem with the attempted na\"ive
abstract interpretation is solved: environments and closures are no
longer mutually recursive.
Unfortunately, continuations still have recursive structure.
We could crudely abstract a continuation into a set of frames, losing
all sense of order, but this would lead to a static analysis lacking
faculties to reason about return-flow: every call would appear to
return to every other call.
A better solution is to refactor continuations as we did environments,
redirecting the recursive structure through the store.
In the next section, we explore a CESK machine with a pointer
refinement for continuations.

\subsection{The CESK\texorpdfstring{$^\star$}{*} machine} 
\label{sec:ceskp}
To untie the recursive structure associated with continuations, we
shift to store-allocated continuations. 
The $\s{Kont}$ component of the machine is replaced by a \emph{pointer
  to} a continuation allocated in the store.  
We term the resulting machine the \CESP{} (control, environment,
store, continuation pointer) machine.
Notice the store now maps to denotable values \emph{and} continuations:
\[
\begin{grammar}
  \state &\in& \State &=& 
  \syn{Exp} \times \s{Env} \times \s{Store} \times \s{Addr}
  \\
  \sto &\in& \Storable &=& \Den \times \Env + \s{Kont}
  \\
  \kappa &\in& \s{Kont} &\produces&
  \mtk \opor \ark(\expr,\env,\addr) \opor \fnk(\den,\env,\addr)\text.
\end{grammar}
\]

The revised machine is defined in
Figure~\ref{fig:cesa} and the initial machine state is defined as:
\[
\inj_{\mathit{CESK^\star}}(\expr) = \langle\expr,\mte,[\addr_0 \mapsto \mtk],\addr_0\rangle\text.
\]

\begin{figure}
\[
\begin{array}{@{}l|r@{}}

\multicolumn{2}{c}{\state \longmapsto_{\mathit{CESK}^\star} \state'
\text{, where }\cont=\store(\addr),\addrnext \notin \dom(\sigma)} \\[1mm]
\hline & \\

\langle\vv, \env, \store, \addr\rangle &
\langle\den,\env', \store, \addr\rangle
\mbox{ where }(\den,\env') = \store(\env(\vv))
\\[1mm]
\langle\appform{\expr_0}{\expr_1}, \env, \store, \addr\rangle &
\langle\expr_0, \env, \store[\addrnext\mapsto \ark(\expr_1, \env, \addr)], \addrnext\rangle
\\[1mm]
\langle\den,\env,\store, \addr\rangle 
\\
\mbox{ if }\cont = \ark(\expr,\env',\addrnextnext) 
&
\langle\expr,\env',\store[\addrnext\mapsto \fnk(\den,\env,\addrnextnext)],\addrnext\rangle
\\
\mbox{ if } \cont = \fnk(\lamform{\vv}{\expr},\env',\addrnextnext) 
&
\langle\expr,\env'[\vv \mapsto \addrnext], \store[\addrnext \mapsto (\den,\env)], \addrnextnext\rangle
\end{array}
\]
\caption{The CESK$^\star$ machine.}
\label{fig:cesa}
\end{figure}
The evaluation function (not shown) is defined along the same lines as
those for the CEK (Section~\ref{sec:cek}) and CESK
(Section~\ref{sec:cesk}) machines.
Like the CESK machine, it is easy to relate the CESK$^\star$ machine
to its predecessor; from corresponding initial configurations, these
machines operate in lock-step:
\begin{lemma}
\label{lem:pointer-equiv}
$\evf_{\mathit{CESK^\star}}(e) \simeq
\evf_{\mathit{CESK}}(e)$.
\end{lemma}

\paragraph{Addresses, abstraction and allocation:}

The CESK$^\star$ machine, as defined in Figure~\ref{fig:cesa},
nondeterministically chooses addresses when it allocates a location in
the store, but because machines are identified up to consistent
renaming of addresses, the transition system remains deterministic.

Looking ahead, an easy way to bound the state-space of this machine is
to bound the set of addresses.\footnote{A finite number of addresses
  leads to a finite number of environments, which leads to a finite
  number of closures and continuations, which in turn, leads to a
  finite number of stores, and finally, a finite number of states.}
But once the store is finite, locations may need to be reused and when
multiple values are to reside in the same location; the store will
have to soundly approximate this by \emph{joining} the values.

In our concrete machine, all that matters about an allocation strategy
is that it picks an unused address.  In the abstracted machine
however, the strategy \emph{may have to re-use previously allocated
  addresses}.  The abstract allocation strategy is therefore crucial
to the design of the analysis---it indicates when finite resources
should be doled out and decides when information should deliberately
be lost in the service of computing within bounded resources.  In
essence, the allocation strategy is the heart of an analysis
(allocation strategies corresponding to well-known analyses are given in
Section~\ref{sec:labelled-acesk}.)

For this reason, concrete allocation deserves a bit more attention in
the machine.  An old idea in program analysis is that dynamically
allocated storage can be represented by the state of the computation
at allocation time~\cite[Section 1.2.2]{dvanhorn:Jones1982Flexible,dvanhorn:midtgaard-07}.  That is, allocation
strategies can be based on a (representation) of the machine history.
These representations are often called \emph{time-stamps}.

A common choice for a time-stamp, popularized by
Shivers~\cite{mattmight:Shivers:1991:CFA}, is to represent the history
of the computation as \emph{contours}, finite strings encoding the
calling context.
We present a concrete machine that uses general time-stamp approach
and is parameterized by a choice of $\tick$ and $\alloc$ functions.
We then instantiate $\tick$ and $\alloc$ to obtain an 
abstract machine for computing a $k$-CFA-style analysis using the
contour approach.


\subsection{The time-stamped CESK\texorpdfstring{$^\star$}{*} machine}
\label{sec:secpt}

The machine states of the time-stamped \CESP{} machine include a
\emph{time} component, which is intentionally left unspecified:
\[
\begin{array}{r@{\;}c@{\;}l}
  \tm,\tmnext &\in &\s{Time}
  \\
  \state &\in& \State =
  \syn{Exp} \times \s{Env} \times \s{Store} \times \s{Addr} \times \s{Time}  
  \text.
\end{array}
\]
The machine is parameterized by the functions:
\begin{align*}
\tick &: \State \rightarrow \s{Time}
&
\alloc &: \State \to \s{Addr}
\text.
\end{align*}
The $\tick$ function returns the next time; the $\alloc$ function
allocates a fresh address for a binding or continuation.
We require of $\tick$ and $\alloc$ that for all $\tm$ and $\state$,
$\tm \sqsubset \tick(\state)$ and $\alloc(\state) \notin \sigma$ where
$\state = \langle \_,\_,\store,\_,\_\rangle$.

The time-stamped \CESP{} machine is defined in Figure~\ref{fig:cesat}.
Note that occurrences of $\state$ on the right-hand side of this
definition are implicitly bound to the state occurring on the
left-hand side.
The initial machine state is defined as:
\[
\inj_{\mathit{CESK^\star_\tm}}(\expr) = \langle\expr,\mte,[\addr_0 \mapsto \mtk],\addr_0,\tm_0\rangle\text.
\]
\begin{figure}
\[
\begin{array}{@{}l@{\;}|@{\;}r@{}}

\multicolumn{2}{c}{\state \longmapsto_{\mathit{CESK}^\star_\tm} \state'
\text{, where }\cont=\store(\addr), \addrnext=\alloc(\state), \tmnext=\tick(\state)} 
\\[1mm]
\hline \\
\langle\vv, \env, \store, \addr,\tm\rangle &
\langle\den,\env', \store, \addr,\tmnext\rangle
\mbox{ where }(\den,\env') = \store(\env(\vv))
\\[1mm]
\langle\appform{\expr_0}{\expr_1}, \env, \store, \addr,\tm\rangle 
&
\langle\expr_0, \env, \store[\addrnext\mapsto \ark(\expr_1, \env, \addr)], \addrnext,\tmnext\rangle
\\[1mm]
\langle\den,\env,\store, \addr,\tm\rangle 
\\
\mbox{if }\cont = \ark(\expr,\env,\addrnextnext) 
&
\langle\expr,\env,\store[\addrnext\mapsto \fnk(\den,\env,\addrnextnext)],\addrnext,\tmnext\rangle
\\
\mbox{if }\cont = \fnk(\lamform{\vv}{\expr},\env',c) 
& 
\langle\expr,\env'[\vv \mapsto b], \store[b \mapsto (\den,\env)], c,\tmnext\rangle
\end{array}
\]
\caption{The time-stamped CESK$^\star$ machine.}
\label{fig:cesat}
\end{figure}

Satisfying definitions for the parameters are:
\begin{gather*}
\s{Time} = \s{Addr} = \mathbb{Z}
\\
\begin{align*}
\addr_0 = \tm_0 &= 0
&
\tick\langle \_, \_, \_, \_, \tm\rangle &= \tm+1
&
\alloc \langle \_, \_, \_, \_, \tm\rangle &= \tm
\text.
\end{align*}
\end{gather*}
Under these definitions, the time-stamped CESK$^\star$ machine
operates in lock-step with the CESK$^\star$ machine, and therefore
with the CESK and CEK machines as well.
\begin{lemma}
\label{lem:time-stamp-equiv}
$\evf_{\mathit{CESK^\star_\tm}}(e) \simeq
\evf_{\mathit{CESK^\star}}(e)$.
\end{lemma}
\noindent
The time-stamped CESK$^\star$ machine forms the basis of our
abstracted machine in the following section.

\subsection{The abstract time-stamped CESK\texorpdfstring{$^\star$}{*} machine}
\label{sec:acespt}

As alluded to earlier, with the time-stamped CESK$^\star$ machine, we
now have a machine ready for direct abstract interpretation via a
single point of approximation: the store.
Our goal is a machine that resembles the
time-stamped CESK$^\star$ machine, but operates over a finite
state-space and it is allowed to be nondeterministic.
Once the state-space is finite, the transitive
closure of the transition relation becomes computable, and this
transitive closure constitutes a static analysis.
Buried in a path through the transitive closure is a (possibly
infinite) traversal that corresponds to the concrete execution of the
program.

The abstracted variant of the time-stamped CESK$^\star$ machine comes from
bounding the address space of the store and the number of times
available.
By bounding these sets, the state-space becomes
finite,\footnote{Syntactic sets like $\syn{Exp}$ are infinite, but
  finite for any given program.}  but for the purposes of soundness,
an entry in the store may be forced to hold several values
simultaneously:
\[
\begin{grammar}
  \astore &\in &\sa{Store} &=& \s{Addr} \parto \Pow{\Storable}
  \text.
\end{grammar}
\]
Hence, stores now map an address to a \emph{set} of
storable values rather than a single value.  These
collections of values model approximation in the analysis.  If a
location in the store is re-used, the new value is joined with the
current set of values.  When a location is dereferenced, the analysis
must consider any of the values in the set as a result of the
dereference.

The abstract time-stamped CESK$^\star$ machine is defined in Figure~\ref{fig:acesat}.
The (non-deterministic) abstract transition relation changes little compared with
the concrete machine.
We only have to modify it to account for the possibility that multiple
storable values (which includes continuations) may reside together in
the store, which we handle by letting the machine
non-deterministically choose a particular value from the set at a
given store location.

\begin{figure}
\[
\begin{array}{@{}l@{\;}|@{\;}r@{}}

\multicolumn{2}{@{}c@{}}{\astate \longmapsto_{\widehat{\mathit{CESK}^\star_\tm}} \astate'
\text{, where } \cont\in\astore(\addr), \addrnext=\aalloc(\astate,\cont), \tmnext=\atick(\astate,\cont)}
\\[2mm]
\hline & \\
\langle\vv, \env, \astore, \addr,\tm\rangle &
\langle \den,\env', \astore, \addr,\tmnext\rangle
\text{ where } (\den,\env') \in \astore(\env(\vv))
\\[1mm]
\langle\appform{\expr_0}{\expr_1}, \env, \astore, \addr,\tm\rangle 
&
\langle\expr_0, \env, \astore \join [\addrnext\mapsto \ark(\expr_1, \env, \addr)], \addrnext,\tmnext\rangle
\\[1mm]
\langle \den,\env,\astore, \addr,\tm\rangle 
\\
\mbox{if }\cont= \ark(\expr,\env',\addrnextnext)
&
\langle\expr,\env',\astore \join [\addrnext\mapsto \fnk(\den,\env,\addrnextnext)],\addrnext,\tmnext\rangle
\\
\mbox{if }\cont=\fnk(\lamform{\vv}{\expr},\env',\addrnextnext)
&
\langle\expr,\env'[\vv \mapsto \addrnext], \astore \join [\addrnext \mapsto (\den,\env)], \addrnextnext,\tmnext\rangle
\end{array}
\]
\caption{The abstract time-stamped CESK$^\star$ machine.}
\label{fig:acesat}
\end{figure}

The analysis is parameterized by abstract variants of the functions
that parameterized the concrete version:
\begin{align*}
\atick &: \aState \times \s{Kont} \rightarrow \s{Time}\text,
&
\aalloc &: \aState \times \s{Kont} \to \s{Addr}
\text.
\end{align*}
In the concrete, these parameters determine allocation and
stack behavior.
In the abstract, they are the arbiters of precision: they determine
when an address gets re-allocated, how many addresses get allocated,
and which values have to share addresses.

Recall that in the concrete semantics, these functions consume
states---not states and continuations as they do here.
This is because in the concrete, a state alone suffices since the
state determines the continuation.
But in the abstract, a continuation pointer within a state may denote a
multitude of continuations; however the transition relation is
defined with respect to the choice of a particular one.
We thus pair states with continuations to encode the choice.

The \emph{abstract} semantics computes the set of reachable states:
\[
\avf_{\widehat{\mathit{CESK}^\star_\tm}}(\expr) = \{ \astate\ |\ \langle\expr,\mte,[\addr_0 \mapsto \mtk],\addr_0,\tm_0\rangle \multistep_{\widehat{\mathit{CESK}^\star_\tm}} \astate \}\text.
\]

\subsection{Soundness and computability}
\label{sec:cesk-soundness}

The finiteness of the abstract state-space ensures decidability.

\begin{theorem}[Decidability of the Abstract CESK$^\star$ Machine]\ \\
$\astate \in \avf_{\widehat{\mathit{CESK}^\star_\tm}}(\expr)$ is decidable.
\end{theorem}
\begin{proof}
  The state-space of the machine is non-recursive with finite sets at the leaves
  on the assumption that addresses are finite.  Hence 
  reachability is decidable since the abstract state-space is finite.
\end{proof}

We have endeavored to evolve the abstract machine gradually so that
its fidelity in soundly simulating the original CEK machine is both
intuitive and obvious.
But to formally establish soundness of the abstract time-stamped
\CESP{} machine, we use an abstraction function, defined in
Figure~\ref{fig:abs-map}, from the state-space of the concrete
time-stamped machine into the abstracted state-space.

\begin{figure}
\begin{align*}
  \absmap(\expr,\env,\store,\addr,\tm) &= (e,\absmap(\env),\absmap(\store),\absmap(\addr),\absmap(\tm)) 
    && \text{[states]}
  \\
  \absmap(\env) &= \lambda \vv . \absmap(\env(\vv)) && \text{[environments]}
\\
  \absmap(\store) &= \lambda \aaddr . \!\!\!\! \bigjoin_{\absmap(\addr) = \aaddr} \!\!\!\! \set{\absmap(\store(\addr))} && \text{[stores]}  
\\
  \absmap(\lamform{\vv}{\expr},\env) &= (\lamform{\vv}{\expr},\absmap(\env)) && \text{[closures]}
  \\
  \absmap(\mtk) &= \mtk && \text{[continuations]}
  \\
  \absmap(\ark(\expr,\env,\addr)) &= \ark(\expr,\absmap(\env),\absmap(\addr))
  \\
  \absmap(\fnk(\den,\env,\addr)) &= \fnk(\den,\absmap(\env),\absmap(\addr))
  \text,
\end{align*}
\caption{The abstraction map, $\absmap :
  \State_{\mathit{CESK}^\star_\tm} \rightarrow
  \aState_{\widehat{\mathit{CESK}^\star_\tm}}$.}
\label{fig:abs-map}
\end{figure}

The abstraction map over times and addresses is defined so that
the parameters $\aalloc$ and $\atick$ are
sound simulations of the parameters $\alloc$
and $\tick$, respectively.
We also define the partial order $(\wt)$ on the abstract state-space
as the natural point-wise, element-wise, component-wise and
member-wise lifting, wherein the partial orders on the sets
$\syn{Exp}$ and $\s{Addr}$ are flat.
Then, we can prove that abstract machine's transition relation
simulates the concrete machine's transition relation.
\begin{theorem}[Soundness of the Abstract CESK$^\star$ Machine]\ \\   
  If $\state \longmapsto_{\mathit{CEK}} \state'$ and
  $\alpha(\state) \wt \astate$, then there exists an abstract state
  $\astate'$, such that $\astate
  \longmapsto_{\widehat{\mathit{CESK}}^\star_\tm} \astate'$ and
  $\alpha(\state') \wt \astate'$.
\end{theorem}
\begin{proof}
  By Lemmas~\ref{lem:store-equiv},~\ref{lem:pointer-equiv},
  and~\ref{lem:time-stamp-equiv}, it suffices to prove soundness with
  respect to $\longmapsto_{\mathit{CESK}^\star_\tm}$.
  Assume $\state \longmapsto_{\mathit{CESK}^\star_\tm} \state'$ and
  $\alpha(\state) \wt \astate$.
  Because $\state$ transitioned, exactly one of the rules from the definition of
  $(\longmapsto_{\mathit{CESK}^\star_\tm})$ applies.
  We split by cases on these rules.
  The rule for the second case is deterministic and follows by calculation.
  For the the remaining (nondeterministic) cases, we must show
  an abstract state exists such that the simulation is
  preserved.
  By examining the rules for these cases, we see that all three hinge
  on the abstract store in $\astate$ soundly approximating the
  concrete store in $\state$, which follows from the assumption that
  $\absmap(\state) \wt \astate$.
\end{proof}

\subsection{A \texorpdfstring{\(\boldsymbol k\)}{k}-CFA-like abstract CESK\texorpdfstring{$^\star$}{*} machine}
\label{sec:labelled-acesk}

In this section, we instantiate the time-stamped CESK$^\star$ machine
to obtain a contour-based machine; this instantiation forms the basis
of a context-sensitive abstract interpreter with polyvariance like
that found in $k$-CFA~\cite{mattmight:Shivers:1991:CFA}.
In preparation for abstraction, we instantiate the time-stamped machine
using labeled call strings.

Inside times, we use \emph{contours} ($\s{Contour}$), which are finite strings of
call site labels that describe the current context:
\[
\begin{grammar}
  \delta &\in & \s{Contour} &\produces&
  \epsilon \opor \lab\delta
  \text.
\end{grammar}
\]

The labeled CESK machine transition relation must appropriately
instantiate the parameters $\tick$ and $\alloc$ to augment the
time-stamp on function call.

Next, we switch to abstract stores and bound the address space by
truncating call string contours to length at most $k$ (for $k$-CFA):
\begin{equation*}
  \delta \in \sa{Contour}_k \text{ iff }\delta \in \s{Contour}\text{ and }\abs{\delta} \leq k\text.
\end{equation*}
Combining these changes, we arrive at the instantiations
for the concrete and abstract machines given in Figure~\ref{fig:kcfa},
where the value $\lfloor\delta\rfloor_k$ is the leftmost $k$ labels of contour $\delta$.

\begin{figure}
\begin{gather*}
\begin{align*}
\s{Time} &= (\syn{Lab} + \bullet) \times \s{Contour}
\\
\s{Addr} &= (\syn{Lab} + \syn{Var}) \times \s{Contour}
\\[1mm]
\tm_0 &= (\bullet,\epsilon)
\\
\tick\langle\vv,\_,\_,\_,\tm\rangle &= \tm
\\
\tick\langle\appform{\expr_0}{\expr_1}^\lab,\_,\_,\_,(\_,\delta)\rangle
&= (\lab,\delta)
\\
\tick\langle\den,\_,\store,\addr,(\lab,\delta)\rangle &= 
\begin{cases}
(\lab,\delta), 
& \mbox{if }\store(\addr) = \ark(\_,\_,\_)
\\
(\bullet,\lab\delta), 
& \mbox{if }\store(\addr) = \fnk(\_,\_,\_)
\end{cases}
\end{align*}
\\[1mm]
\begin{align*}
\alloc(\langle\appform{\expr_0^\lab}{\expr_1},\_,\_,\_,(\_,\delta)\rangle)
&= (\lab,\delta)
\\
\alloc(\langle\den,\_,\store,\addr,(\_,\delta)\rangle)
&= (\lab,\delta) \mbox{ if }  \store(\addr) = \ark(\expr^\lab,\_,\_)
\\[1mm]
\alloc(\langle\den,\_,\store,\addr,(\_,\delta)\rangle)
&= (\vv,\delta) \mbox{ if } \store(\addr) = \fnk(\lamform{\vv}{\expr},\_,\_)
\end{align*}
\end{gather*}
\begin{gather*}
\begin{align*}
\atick(\langle\vv,\_,\_,\_,\tm\rangle,\cont) &= \tm
\\
\atick(\langle\appform{\expr_0}{\expr_1}^\lab,\_,\_,\_,(\_,\delta)\rangle,\cont)
&= (\lab,\delta)
\\
\atick(\langle\den,\_,\astore,\addr,(\lab,\delta)\rangle,\cont) &= 
\begin{cases}
(\lab,\delta), 
& \mbox{if }\cont= \ark(\_,\_,\_)
\\
(\bullet,\lfloor\lab\delta\rfloor_k), 
& \mbox{if }\cont= \fnk(\_,\_,\_)
\end{cases}
\end{align*}
\\[1mm]
\begin{align*}
\aalloc(\langle\appform{\expr_0^\lab}{\expr_1},\_,\_,\_,(\_,\delta)\rangle, \cont)
&= (\lab,\delta)
\\
\aalloc(\langle\den,\_,\astore,\addr,(\_,\delta)\rangle, \cont)
&= (\lab,\delta) \mbox{ if } \cont = \ark(\expr^\lab,\_,\_)
\\[1mm]
\aalloc(\langle\den,\_,\astore,\addr,(\_,\delta)\rangle, \cont)
&= (\vv,\delta) \mbox{ if } \cont = \fnk(\lamform{\vv}{\expr},\_,\_)
\end{align*}
\end{gather*}
\caption{Instantiation for $k$-CFA machine.}
\label{fig:kcfa}
\end{figure}


\paragraph{Comparison to \texorpdfstring{$\boldsymbol k$}{k}-CFA:}

We say ``$k$-CFA-like'' rather than ``$k$-CFA'' because there are
distinctions between the machine just described and $k$-CFA:
\begin{enumerate}

\item $k$-CFA focuses on ``what flows where''; the ordering between
  states in the abstract transition graph produced by our machine
  produces ``what flows where \emph{and when}.''

\item Standard presentations of $k$-CFA implicitly inline a global
  approximation of the store into the
  algorithm~\cite{mattmight:Shivers:1991:CFA}; ours uses one store per
  state to increase precision at the cost of complexity.  
  In terms of our framework, the lattice through which classical
  $k$-CFA ascends is $\Pow{\s{Exp} \times \s{Env} \times \s{Addr}}
  \times \sa{Store}$, whereas our analysis ascends the lattice
  $\Pow{\s{Exp} \times \s{Env} \times \sa{Store} \times \s{Addr}}$.
  We can explicitly inline the store to
  achieve the same complexity, as shown in Section~\ref{sec:widening}.

\item On function call, $k$-CFA merges argument values together with
  previous instances of those arguments from the same context; our
  ``minimalist'' evolution of the abstract machine takes a
  higher-precision approach: it forks the machine for each argument
  value, rather than merging them immediately.

\item $k$-CFA does not recover explicit information about stack
  structure; our machine contains an explicit model of the stack for
  every machine state.

\end{enumerate}

\section{Analyzing by-need with Krivine's machine}
\label{sec:krivine}
Even though the abstract machines of the prior section have advantages
over traditional CFAs, the approach we took (store-allocated
continuations) yields more novel results when applied in a different
context: a lazy variant of Krivine's machine.
That is, we can construct an abstract interpreter that both analyzes
and exploits laziness.
Specifically, we present an abstract analog to a lazy and properly
tail-recursive variant of Krivine's
machine~\cite{dvanhorn:Krivine1985Un,dvanhorn:Krivine2007Callbyname}
derived by Ager, Danvy, and
Midtgaard~\cite{dvanhorn:Ager2004Functional}.  
The derivation from Ager \etal's machine to the abstract interpreter
follows the same outline as that of Section~\ref{sec:cek-to-acesk}: we
apply a pointer refinement by store-allocating continuations and
carry out approximation by bounding the store.

The by-need variant of Krivine's machine considered here uses the
common implementation technique of store-allocating thunks and forced
values.  When an application is evaluated, a thunk is created that
will compute the value of the argument when forced.  When a variable
occurrence is evaluated, if it is bound to a thunk, the thunk is
forced (evaluated) and the store is updated to the result.  Otherwise
if a variable occurrence is evaluated and bound to a forced value,
that value is returned.

Storable values include delayed computations (thunks)
$\delayed(\expr,\env)$, and computed values $\computed(\den,\env)$,
which are just tagged closures.
There are two continuation constructors: $\kok(\addr,\cont)$ is
induced by a variable occurrence whose binding has not yet
been forced to a value.  The address $\addr$ is where we want to write
the given value when this continuation is invoked. The other:
$\ktk(\addr,\cont)$ is induced by an application expression, which
forces the operator expression to a value.  The address $\addr$ is the
address of the argument.

The concrete state-space is defined as follows and the transition
relation is defined in Figure~\ref{fig:lk}:
\[
\begin{grammar}
  \state &\in& \State &=& 
  \syn{Exp} \times \s{Env} \times \s{Store} \times \s{Kont} 
  \\
  \sto &\in & \Storable &\produces & \delayed(\expr,\env) \opor \computed(\den,\env) 
  \\
  \cont &\in &\s{Kont}&\produces &\mtk \opor \kok(\addr,\cont) \opor \ktk(\addr,\cont)
\end{grammar}
\]

\begin{figure}
\[
\begin{array}{@{}l@{\qquad}|@{\quad}r@{}}

\multicolumn{2}{c}{\state \longmapsto_{\mathit{LK}} \state'} \\[1mm]
\hline & \\

\langle\vv, \env, \store, \cont\rangle &
\\
\mbox{if }\store(\rho(\vv)) = \delayed(\expr,\env')
&
\langle \expr, \env', \store, \kok(\rho(\vv),\cont)\rangle
\\
\mbox{if }\store(\rho(\vv)) = \computed(\den,\env')
&
\langle\den, \env', \store, \cont\rangle
\\[1mm]
\langle\appform{\expr_0}{\expr_1}, \env, \store, \cont\rangle 
&
\langle\expr_0, \env, \store[\addr \mapsto \delayed(\expr_1,\env)], \ktk(\addr,\cont) \rangle
\\
&
\mbox{where }\addr \notin \dom(\store)
\\[1mm]
\langle\den,\env,\store,\kok(\addr,\cont)\rangle
&
\langle\den,\env,\store[\addr\mapsto \computed(\den,\env)],\cont\rangle
\\[1mm]
\langle\lamform{\vv}{\expr},\env,\store,\ktk(\addr,\cont)\rangle
&
\langle\expr,\env[\vv\mapsto\addr],\store,\cont\rangle
\\
\end{array}
\]
\caption{The LK machine.}
\label{fig:lk}
\end{figure}

When the control component is a variable, the machine looks up its
stored value, which is either computed or delayed.  If delayed, a
$\kok$ continuation is pushed and the frozen expression is put in
control.  If computed, the value is simply returned.  When a value is
returned to a $\kok$ continuation, the store is updated to reflect the
computed value.  When a value is returned to a $\ktk$ continuation,
its body is put in control and the formal parameter is bound to the
address of the argument.

We now refactor the machine to use store-allocated continuations;
storable values are extended to include continuations:
\[
\begin{grammar}
  \state &\in& \State &=& 
  \syn{Exp} \times \s{Env} \times \s{Store} \times \s{Addr} 
  \\
  \sto &\in & \Storable &\produces & \delayed(\expr,\env) \opor \computed(\den,\env) \opor \cont 
  \\
  \cont &\in &\s{Kont}&\produces &\mtk \opor \kok(\addr,\addr) \opor \ktk(\addr,\addr)
  \text.
\end{grammar}
\]
It is straightforward to perform a pointer-refinement of the LK
machine to store-allocate continuations as done for the CESK machine
in Section~\ref{sec:ceskp} and
observe the lazy variant of Krivine's machine
and its pointer-refined counterpart (not shown) operate in lock-step:
\begin{lemma}
$\evf_{\mathit{LK}}(\expr) \simeq \evf_{\mathit{LK}^\star}(\expr)$.
\end{lemma}

\begin{figure}
\[
\begin{array}{@{}l|r@{}}

\multicolumn{2}{c}{\astate \longmapsto_{\widehat{\mathit{LK}^\star_\tm}} \astate'
\mbox{, where }\cont \in \astore(\addr), \addrnext=\aalloc(\astate,\cont), \tmnext=\atick(\astate,\cont)}
 \\[2mm]
\hline & \\

\langle\vv, \env, \astore, \addr, \tm\rangle 
&
\langle \expr, \env', \astore \sqcup [\addrnext\mapsto \kok(\rho(\vv),\addr)], \addrnext,\tmnext\rangle
\\
\mbox{if }\astore(\rho(\vv)) \ni \delayed(\expr,\env')
\\[1mm]
\langle\vv, \env, \astore, \addr,\tm\rangle &
\langle\den, \env', \astore, \addr,\tmnext\rangle
\\
\mbox{if }\astore(\rho(\vv)) \ni \computed(\den,\env')
\\[1mm]
\langle\appform{\expr_0}{\expr_1}, \env, \astore, \addr,\tm\rangle
&
\langle\expr_0, \env, \astore', \addrnext,\tmnext\rangle
\\
&
\mbox{where }\addrnextnext = \aalloc(\astate,\cont),
\\
&
\astore' = \astore\sqcup[\addrnextnext \mapsto \delayed(\expr_1,\env), \addrnext\mapsto \ktk(\addrnextnext,\addr)]
\\[1mm]
\langle\den,\env,\astore,\addr,\tm\rangle
&
\langle\den,\env',\astore\join[\addr'\mapsto \computed(\den,\env)],\addrnextnext,\tmnext\rangle
\\
\mbox{if }\cont= \kok(\addr',\addrnextnext)
\\[1mm]
\langle\lamform{\vv}{\expr},\env,\astore,\addr,\tm\rangle
&
\langle\expr,\env'[\vv\mapsto\addr'],\astore,\addrnextnext,\tmnext\rangle
\\
\mbox{if }\cont= \ktk(\addr',\addrnextnext)
\end{array}
\]
\caption{The abstract LK$^\star$ machine.}
\label{fig:lka}
\end{figure}

After threading time-stamps through the machine as done in
Section~\ref{sec:secpt} and defining $\atick$ and $\aalloc$
analogously to the definitions given in Section~\ref{sec:acespt}, the
pointer-refined machine abstracts directly to yield the abstract
LK$^\star$ machine in Figure~\ref{fig:lka}.

The abstraction map 
for this machine is a straightforward
structural abstraction similar to that given in
Section~\ref{sec:cesk-soundness} (and hence omitted).
The abstracted machine is sound with respect to the LK$^\star$
machine, and therefore the original LK machine.

\begin{theorem}[Soundness of the Abstract LK$^\star$ Machine]\ \\
  If $\state \longmapsto_{\mathit{LK}} \state'$ and
  $\alpha(\state) \wt \astate$, then there exists an abstract state
  $\astate'$, such that $\astate
  \longmapsto_{\widehat{\mathit{LK}^\star_\tm}} \astate'$ and
  $\alpha(\state') \wt \astate'$.
\end{theorem}

\paragraph{Optimizing the machine through specialization:}
Ager \etal~optimize the LK machine by specializing application
transitions.  When the operand of an application is a variable, no
delayed computation needs to be constructed, thus ``avoiding the
construction of space-leaky chains of thunks.''  Likewise, when the
operand is a $\lambda$-abstraction, ``we can store the corresponding
closure as a computed value rather than as a delayed computation.''
Both of these optimizations, which conserve valuable abstract
resources, can be added with no trouble, as shown in Figure~\ref{fig:alko}.

\begin{figure}[h]
\[
\begin{array}{@{}l|r@{}}

\multicolumn{2}{c}{\astate \longmapsto_{\widehat{\mathit{LK}^\star}} \astate'
\text{, where }\cont\in\astore(\addr), 
\addrnext=\aalloc(\astate,\cont),
\tmnext=\atick(\astate,\cont)} \\[2mm]
\hline & \\

\langle\appform{\expr}{\vv}, \env, \astore, \addr,\tm\rangle 
&
\langle\expr, \env, \astore \join[\addrnext\mapsto
  \ktk(\env(\vv),\addr)], \addrnext,\tmnext\rangle
\\[1mm]
\langle\appform{\expr}{\den}, \env, \astore, \addr,\tm\rangle 
&
\langle\expr_0, \env, \astore \join [\addrnext \mapsto \computed(\den,\env), \addrnextnext\mapsto \ktk(\addrnext,\addr)], \addrnextnext,\tmnext\rangle
\\
&
\mbox{where }\addrnextnext=\aalloc(\astate,\cont)
\end{array}
\]
\caption{The abstract optimized LK$^\star$ machine.}
\label{fig:alko}
\end{figure}

\paragraph{Varying the machine through postponed thunk creation:}
Ager \etal~also vary the LK machine by postponing the construction of
a delayed computation from the point at which an application is the
control string to the point at which the operator has been evaluated
and is being applied.  The $\ktk$ continuation is modified to hold,
rather than the address of a delayed computation, the constituents of
the computation itself:
\[
\begin{grammar}
  \cont &\in &\s{Kont}&\produces &\mtk \opor \kok(\addr,\addr) \opor \ktk(\expr,\env,\addr)
  \text.
\end{grammar}
\]
The transitions for applications and functions
are replaced with those in Figure~\ref{fig:lkav}.
This allocates thunks when a function is applied,
rather than when the control string is an application.

\begin{figure}
\[
\begin{array}{@{}l@{\quad}|@{\quad}r@{}}

\multicolumn{2}{c}{\astate \longmapsto_{\widehat{\mathit{LK}'^\star}} \astate'
\text{, where }\cont\in\astore(\addr), 
\addrnext=\aalloc(\astate,\cont),
\tmnext=\atick(\astate,\cont)} \\[2mm]
\hline & \\

\langle\appform{\expr_0}{\expr_1}, \env, \astore, \addr\rangle 
&
\langle\expr_0, \env, \astore\sqcup[\addrnext\mapsto \ktk(\expr_1,\env,\addr)], \addrnext\rangle
\\[1mm]
\langle\lamform{\vv}{\expr},\env,\astore,\addr\rangle
&
\langle\expr,\env[\vv\mapsto\addrnext],\astore\sqcup[\addrnext\mapsto \delayed(\expr',\env')],\addrnextnext\rangle
\\
\mbox{ if }\cont=\ktk(\expr',\env',\addrnextnext)
\end{array}
\]
\caption{The abstract thunk postponing LK$^\star$ machine.}
\label{fig:lkav}
\end{figure}

As Ager~\etal\ remark, each of these variants gives rise to an
abstract machine.  From each of these machines, we are able to
systematically derive their abstractions.

\section{State and control}
\label{sec:realistic-features}

We have shown that store-allocated continuations make abstract
interpretation of the CESK machine and a lazy variant of Krivine's
machine straightforward.
In this section, we want to show that the tight correspondence between
concrete and abstract persists after the addition of language
features such as conditionals, side effects, exceptions and
continuations.
We tackle each feature, and present the additional machinery required
to handle each one.
In most cases, the path from a canonical concrete machine to
pointer-refined abstraction of the machine is so simple we only show
the abstracted system.
In doing so, we are arguing that this abstract machine-oriented
approach to abstract interpretation represents a flexible and viable
framework for building abstract interpreters.

\subsection{Conditionals, mutation, and control}

To handle conditionals, we extend the language with a new syntactic
form, $\ifform{\expr}{\expr}{\expr}$, and introduce a base value
$\schfalse$, representing false.
Conditional expressions induce a new continuation form:
$\mathbf{if}(\expr'_0,\expr'_1,\env,\addr)$, which represents the
evaluation context $E[\ifform{[\;]}{\expr_0}{\expr_1}]$ where $\env$
closes $\expr'_0$ to represent $\expr_0$, $\env$ closes $\expr'_1$ to
represent $\expr_1$, and $\addr$ is the address of the representation
of $E$.

Side effects are fully amenable to our approach;
we introduce Scheme's \texttt{set!} for mutating variables using the
$\appform{{\tt set!}\;}{\vv\; \expr}$ syntax.  The \texttt{set!}  form
evaluates its subexpression $\expr$ and assigns the value to the
variable $\vv$.  Although \texttt{set!} expressions are evaluated for
effect, we follow Felleisen \etal~and specify \texttt{set!}
expressions evaluate to the value of $\vv$ before it was
mutated~\cite[page 166]{dvanhorn:Felleisen2009Semantics}.  The
evaluation context $E[\appform{{\tt set!}\;}{\vv\; [\;]}]$ is
represented by $\mathbf{set}(\addr_0,\addr_1)$, where $\addr_0$ is the
address of $\vv$'s value and $\addr_1$ is the address of the
representation of $E$.

First-class control is introduced by adding a new base value {\tt
  callcc} which reifies the continuation as a new kind of applicable
value.  Denoted values are extended to include representations of
continuations.  Since continuations are store-allocated, we choose to
represent them by address.  When an address is applied, it represents
the application of a continuation (reified via {\tt callcc}) to a
value.  The continuation at that point is discarded and the applied
address is installed as the continuation.

The resulting grammar is:
\[
\begin{grammar}
 \expr &\in& \syn{Exp} &\produces& \dots \opor \ifform{\expr}{\expr}{\expr}
 \opor \appform{{\tt set!}\;}{\vv\; \expr}
\\
  \kappa &\in& \s{Kont} &\produces&
  \dots \opor \mathbf{if}(\expr,\expr,\env,\addr) \opor \mathbf{set}(\addr,\addr)
\\
  \den &\in& \Den &\produces& \dots \opor \schfalse \opor  {\tt callcc} \opor \addr
  \text.
\end{grammar}
\]
We show only the abstract transitions, which result from
store-allocating continuations, time-stamping, and abstracting the
concrete transitions for conditionals, mutation, and control.
The first three machine transitions deal with conditionals; here we
follow the Scheme tradition of considering all non-false values as
true.
The fourth and fifth transitions deal with mutation.

\begin{figure}
\[
\begin{array}{@{}l@{\;}|@{\;}r@{}}

\multicolumn{2}{c}{\astate \longmapsto_{\widehat{\mathit{CESK}^\star_\tm}} \astate'
\text{, where } 
\cont \in \astore(\addr), 
\addrnext = \aalloc(\astate,\cont),
\tmnext = \atick(\astate,\cont)
 } \\[2mm]
\hline & \\

\langle\ifform{\expr_0}{\expr_1}{\expr_2},\env,\astore,\addr,\tm\rangle
&
\langle\expr_0,\env,\astore \sqcup [\addrnext\mapsto \mathbf{if}(\expr_1,\expr_2,\env,\addr)],\addrnext,\tmnext\rangle
\\[1mm]
\langle\schfalse,\env,\astore,\addr,\tm\rangle
&
\langle\expr_1,\env',\astore,\addrnextnext,\tmnext\rangle
\\
\mbox{if }\cont=
   \mathbf{if}(\expr_0,\expr_1,\env',\addrnextnext)
\\[1mm]
\langle\den,\env,\astore,\addr,\tm\rangle
&
\langle\expr_0,\env',\astore,\addrnextnext,\tmnext\rangle
\\
\mbox{if }\cont=
   \mathbf{if}(\expr_0,\expr_1,\env',\addrnextnext),
\\
\text{and } 
\den\neq\schfalse
\\[1mm]
\langle\appform{{\tt set!}\;}{\vv\; \expr},\env,\astore,\addr,\tm\rangle
&
\langle\expr,\env,\astore \sqcup [\addrnext\mapsto \mathbf{set}(\env(\vv), \addr)],\addrnext,\tmnext\rangle
\\[1mm]
\langle\den,\env,\astore,\addr,\tm\rangle
&
\langle\den',\env,\astore \join [\addr' \mapsto \den],\addrnextnext,\tmnext\rangle
\\
\mbox{ if }\cont=
   \mathbf{set}(\addr',\addrnextnext) 
& \mbox{where }\den' \in \astore(\addr')
\\[1mm]
\langle \lamform{\vv}{\expr},\env,\astore,\addr,\tm\rangle 
&
\langle \expr,\env[\vv\mapsto \addrnext],\astore \join [\addrnext\mapsto \addrnextnext],\addrnextnext,\tmnext\rangle
\\
\mbox{if }\cont= \fnk({\tt callcc},\env',\addrnextnext)
&
\mbox{where }\addrnextnext=\aalloc(\astate,\cont)
\\[1mm]
\langle \addrnextnext,\env,\astore,\addr,\tm\rangle 
&
\langle \addr,\env,\astore,\addrnextnext,\tmnext\rangle 
\\
\mbox{if }\cont= \fnk({\tt callcc},\rho',\addr')
\\[1mm]
\langle \den,\env,\astore,\addr,\tm\rangle 
&
\langle \den,\env,\astore,\addrnextnext,\tmnext\rangle 
\\
\mbox{if }\cont= \fnk(\addrnextnext,\env',\addr')
\end{array}
\]
\caption{The abstract extended CESK$^\star$ machine.}
\end{figure}

The remaining three transitions deal with first-class control.
In the first of these, {\tt callcc} is being applied to a closure
value $\den$.  The value $\den$ is then ``called with the current
continuation'', \ie, $\den$ is applied to a value that represents the
continuation at this point.
In the second, {\tt callcc} is being applied to a continuation
(address).  When this value is applied to the reified continuation, it
aborts the current computation, installs itself as the current
continuation, and puts the reified continuation ``in the hole''.
%
Finally, in the third, a continuation is being applied; $\addrnextnext$ gets
thrown away, and $\den$ gets plugged into the continuation $\addrnext$.

In all cases, these transitions result from pointer-refinement,
time-stamping, and abstraction of the usual machine
transitions.

\subsection{Exceptions and handlers}

To analyze exceptional control flow, we extend the CESK machine with a
register to hold a stack of exception handlers.  
This models a reduction semantics in which we have two additional kinds of
evaluation contexts:
\[
\begin{grammar}
 && E &\produces& [\;] \opor \appform{E}{\expr} \opor \appform{\den}{E}
\opor \catchform{E}{\den}
\\
 && F &\produces& [\;] \opor \appform{F}{\expr} \opor \appform{\den}{F}
\\
&& H &\produces& [\;] \opor H[F[\catchform{H}{\den}]]\text,
\end{grammar}
\]
and the additional, context-sensitive, notions of reduction:
\begin{align*}
\catchform{E[\throwform{\den}]}{\den'} &\rightarrow \appform{\den'}{\den}\text,
&
\catchform{\den}{\den'} &\rightarrow \den\text.
\end{align*}
$H$ contexts represent a stack of exception handlers, while $F$
contexts represent a ``local'' continuation, \ie, the rest of the
computation (with respect to the hole) up to an enclosing handler, if
any. $E$ contexts represent the entire rest of the computation,
including handlers.

The language is extended with expressions for raising and
catching exceptions.  A new kind of continuation is introduced to
represent a stack of handlers.  In each frame of the stack, there is a
procedure for handling an exception and a (handler-free) continuation:
\[
\begin{grammar}
  \expr &\in& \syn{Exp} &\produces& 
  \dots \opor \throwform{\den}
  \opor \catchform{\expr}{\lamform{\vv}{\expr}}
\\
  \handls &\in& \s{Handl} &\produces& \mtk \opor \handk(\den,\env,\cont,\handls)
\end{grammar}
\]
An $\handls$ continuation represents a stack of exception handler
contexts, \ie, $\handk(\den',\env,\cont,\handls)$
represents $H[F[\catchform{[\;]}{\den}]]$, where
$\handls$ represents $H$, $\cont$ represents $F$, and $\env$ closes
$\den'$ to represent $\den$.

The machine includes all of the transitions of the CESK machine
extended with a $\handls$ component; these transitions are omitted
for brevity.
The additional transitions are given in Figure~\ref{fig:ceshk}.
This presentation is based on a textbook treatment of exceptions
and handlers~\cite[page
  135]{dvanhorn:Felleisen2009Semantics}.\footnote{To be precise,
  Felleisen \etal~present the CHC machine, a substitution based
  machine that uses evaluation contexts in place of continuations.
  Deriving the CESHK machine from it is an easy exercise.}
The initial configuration is given by:
\[
\inj_{\mathit{CESHK}}(\expr) = \langle\expr,\mte,\mts,\mtk,\mtk\rangle\text.
\]

\begin{figure}
\[
\begin{array}{@{}l@{\;}|@{\;}r@{}}

\multicolumn{2}{c}{\state \longmapsto_{\mathit{CESHK}} \state'} \\[1mm]
\hline & \\

\langle \den,\env,\store,\handk(\den',\env',\cont,\handls),\mtk\rangle 
&
\langle \den,\env,\store,\handls,\cont\rangle
\\[1mm]
\multicolumn{2}{@{}l}{%
\langle\throwform{\den},\env,\store,\handk(\lamform{\vv}{\expr},\env',\cont',\handls),\cont\rangle
}
\\
&
\langle \expr, \env'[\vv\mapsto \addr],\store[\addr\mapsto (\den,\env)],\handls,\cont'\rangle
\\
& \mbox{where }\addr\notin\dom(\store)
\\[1mm]
\langle \catchform{\expr}{\den},\env,\store,\handls,\cont\rangle
&
\langle \expr,\env,\store,\handk(\den,\env,\cont,\handls),\mtk\rangle
\end{array}
\]
\caption{The CESHK machine.}
\label{fig:ceshk}
\end{figure}

In the pointer-refined machine,
the grammar of handler continuations changes to the following:
\[
\begin{grammar}
  \handls &\in& \s{Handl} &\produces&
  \mtk \opor \handk(\den,\env,\addr,\haddr)\text,
\end{grammar}
\]
where $\haddr$ is used to range over addresses pointing to handler
continuations.
The notation $\addr_\mtk$ means $\addr$ such that $\store(\addr) =
\mtk$ in concrete case and $\mtk \in \astore(\addr)$ in the abstract,
where the intended store should be clear from context.
The pointer-refined machine is given in Figure~\ref{fig:cesha}.

\begin{figure}
\[
\begin{array}{@{}l@{\;}|@{\;}r@{}}

\multicolumn{2}{c}{\state \longmapsto_{\mathit{CESHK}^\star} \state'
\text{, where }
\handls= \store(\haddr), \cont= \store(\addr),
\addrnext \notin \dom(\store)
} 
\\[1mm]
\hline & \\

\langle \den,\env,\store,\haddr,\addr\rangle 
&
\langle \den,\env,\store,\haddr',\addr'\rangle
\\
\mbox{if }\handls = \handk(\den',\env',\addr',\haddr'),
\\
\text{and }\cont = \mtk
\\[1mm]
\langle \throwform{\den},\env,\store,\haddr,\addr\rangle 
&
\langle \expr, \env'[\vv\mapsto \addrnext],\store[\addrnext\mapsto (\den,\env)],\haddr',\addr'\rangle
\\
\mbox{if }\handls = \handk(\lamform{\vv}{\expr},\env',\addr',\haddr')
\\[1mm]
\langle \catchform{\expr}{\den},\env,\store,\haddr,\addr\rangle
&
\langle \expr,\env,\store[\addrnext\mapsto\handk(\den,\env,\addr,\haddr)],\addrnext,\addr_\mtk\rangle
\end{array}
\]
\caption{The CESHK$^\star$ machine.}
\label{fig:cesha}
\end{figure}

After threading time-stamps through the machine as done in
Section~\ref{sec:secpt}, the machine abstracts as usual
to obtain the machine in Figure~\ref{fig:acesaat}.
The only unusual step in the derivation is to observe that some
machine transitions rely on a choice of \emph{two} continuations from
the store; a handler and a local continuation.  Analogously to
Section~\ref{sec:acespt}, we extend $\atick$ and $\aalloc$ to take two
continuation arguments to encode the choice:
\begin{align*}
\atick &: \aState \times \s{Handl} \times \s{Kont} \rightarrow \s{Time}\text,
\\
\aalloc &: \aState \times \s{Handl} \times \s{Kont} \to \s{Addr}
\text.
\end{align*}

\begin{figure}
\[
\begin{array}{@{}l@{\quad}|@{\qquad\quad\qquad}r@{}}

\multicolumn{2}{@{}c@{}}{\astate \longmapsto_{\widehat{\mathit{CESHK}^\star_\tm}} \astate'
\text{, where }
\handls\in \astore(\haddr), \cont\in \astore(\addr),
\addrnext=\aalloc(\astate,\handls,\cont),
} 
\\
\multicolumn{2}{@{}c@{}}{
\tmnext =\atick(\astate,\handls,\cont)}
\\[2mm]
\hline & \\

\langle \den,\env,\astore,\haddr,\addr,\tm\rangle 
&
\langle \den,\env,\astore,\haddr',\addr',\tmnext\rangle
\\
\mbox{if }\handls= \handk(\den',\env',\addr',\haddr'),
\\
\text{and }\cont= \mtk
\\[1mm]
\langle \throwform{\den},\env,\astore,\haddr,\addr,\tm\rangle 
\\
\mbox{if }\handls = \handk(\lamform{\vv}{\expr},\env',\addr',\haddr')
\\
\multicolumn{2}{r@{}}{ 
\langle \expr, \env'[\vv\mapsto \addrnext],\astore\sqcup[\addrnext\mapsto(\den,\env)],\haddr',\addr',\tmnext\rangle
}
\\[1mm]
\langle \catchform{\expr}{\den},\env,\astore,\haddr,\addr,\tm\rangle
\\
\multicolumn{2}{r@{}}{
\
\langle \expr,\env,\astore\sqcup[\addrnext\mapsto\handk(\den,\env,\addr,\haddr)],\addrnext,\addr_\mtk,\tmnext\rangle}
\end{array}
\]
\caption{The abstract CESHK$^\star$ machine.}
\label{fig:acesaat}
\end{figure}

\section{Abstract garbage collection}
\label{sec:garbage-collection}

Garbage collection determines when a store location has become
unreachable and can be re-allocated.
This is significant in the abstract semantics because an address may be
allocated to multiple values due to finiteness of the address space.
Without garbage collection, the values allocated to this common
address must be joined, introducing imprecision in the analysis (and
inducing further, perhaps spurious, computation).
By incorporating garbage collection in the abstract semantics, the
location may be proved to be unreachable and safely \emph{overwritten}
rather than joined, in which case no imprecision is introduced.

Like the rest of the features addressed in this paper, we can
incorporate abstract garbage collection into our static analyzers by a
straightforward pointer-refinement of textbook accounts of concrete
garbage collection, followed by a finite store abstraction.

Concrete garbage collection is defined in terms of a GC machine that
computes the reachable addresses in a store~\cite[page
  172]{dvanhorn:Felleisen2009Semantics}:
\[
\begin{array}{l}
\langle\mathcal{G},\mathcal{B},\store\rangle 
\longmapsto_{\mathit{GC}}
\langle(\mathcal{G}\cup \liveloc_\store(\store(\addr)) \setminus
(\mathcal{B}\cup\{\addr\})), \mathcal{B}\cup\{\addr\}, \store\rangle
\\
\mbox{if }\addr \in \mathcal{G}\text.
\end{array}
\]
This machine iterates over a set of reachable but unvisited ``grey''
locations $\mathcal{G}$.  On each iteration, an element is removed and
added to the set of reachable and visited ``black'' locations
$\mathcal{B}$.  Any newly reachable and unvisited locations, as
determined by the ``live locations'' function $\liveloc_\store$, are
added to the grey set.  When there are no grey locations, the black
set contains all reachable locations.  Everything else is garbage.

The live locations function computes a set of locations which may be
used in the store.  
 Its definition will vary based on the particular
machine being garbage collected, but the definition appropriate for
the CESK$^\star$ machine of Section~\ref{sec:ceskp} is
\begin{align*}
\liveloc_\store(\expr) &= \emptyset
\\
\liveloc_\store(\expr,\env) &= \liveloc_\store(\restrict{\env}{\fv(\expr)})
\\
\liveloc_\store(\env) &= \mathit{rng}(\env)
\\
\liveloc_\store(\mtk) &= \emptyset
\\
\liveloc_\store(\fnk(\den,\env,\addr)) &= \{\addr\} \cup \liveloc_\store(\den,\env) \cup \liveloc_\store(\store(\addr))
\\
\liveloc_\store(\ark(\expr,\env,\addr)) &= \{\addr\}\cup \liveloc_\store(\expr,\env) \cup \liveloc_\store(\store(\addr))\text.
\end{align*}
We write $\restrict\env{\fv(\expr)}$ to
mean $\env$ restricted to the domain of free variables in $\expr$.
We assume the least-fixed-point solution in the calculation of the
function $\liveloc$ in cases where it recurs on itself.

The pointer-refinement of the machine requires parameterizing the
$\liveloc$ function with a store used to resolve pointers to
continuations. 
A nice consequence of this parameterization is that we can re-use
$\liveloc$ for \emph{abstract garbage collection} by supplying it an
abstract store for the parameter.
Doing so only necessitates extending $\liveloc$ to the case of sets of
storable values:
\begin{align*}
\liveloc_\store(S) &= \bigcup_{s\in S} \liveloc_\store(s)
\end{align*}

The CESK$^\star$ machine incorporates garbage collection by a
transition rule that invokes the GC machine as a subroutine to remove
garbage from the store (Figure~\ref{fig:gc}).
The garbage collection transition introduces non-determinism to the
CESK$^\star$ machine because it applies to any machine state and thus
overlaps with the existing transition rules.
The non-determinism is interpreted as leaving the choice of
\emph{when} to collect garbage up to the machine.

The abstract CESK$^\star$ incorporates garbage collection by the
\emph{concrete garbage collection transition}, \ie, we re-use the
definition in Figure~\ref{fig:gc} with an abstract store, $\astore$,
in place of the concrete one.
Consequently, it is easy to verify abstract garbage collection
approximates its concrete counterpart.

\begin{figure}
\[
\begin{array}{@{}l@{\qquad}|@{\qquad}r@{}}

\multicolumn{2}{c}{\state \longmapsto_{\mathit{CESK}^\star} \state'} \\[1mm]
\hline & \\

\langle \expr,\env,\store,\addr\rangle
\qquad\qquad\qquad
&
\langle \expr,\env,\{\langle\addrnext,\store(\addrnext)\rangle\ |\ \addrnext \in \mathcal{L}\},\addr\rangle
\\
\multicolumn{2}{@{}l}{
\mbox{if }\langle\liveloc_\store(\expr,\env) \cup \liveloc_\store(\store(\addr)),\{\addr\},\store\rangle \multistep_{\mathit{GC}} \langle\emptyset,\mathcal{L},\store\rangle}
\end{array}
\]
\caption{The GC transition for the CESK$^\star$ machine.}
\label{fig:gc}
\end{figure}

The CESK$^\star$ machine may collect garbage at any point in the
computation, thus an abstract interpretation must soundly approximate
\emph{all possible choices} of when to trigger a collection, which 
the abstract CESK$^\star$ machine does correctly.
This may be a useful analysis \emph{of} garbage collection, however it
fails to be a useful analysis \emph{with} garbage collection: for
soundness, the abstracted machine must consider the case in which
garbage is never collected, implying no storage is reclaimed to
improve precision.

However, we can leverage abstract garbage collection to reduce the
state-space explored during analysis and to improve precision and
analysis time.
This is achieved (again) by considering properties of the
\emph{concrete} machine, which abstract directly; in this case,
we want the concrete machine to deterministically collect garbage.
Determinism of the CESK$^\star$ machine is restored by defining the
transition relation as a non-GC transition (Figure~\ref{fig:cesa})
followed by the GC transition (Figure~\ref{fig:gc}).
This state-space of this concrete machine is ``garbage free'' and
consequently the state-space of the abstracted machine is ``abstract
garbage free.''

In the concrete semantics, a nice consequence of this property is that
although continuations are allocated in the store, they are
deallocated as soon as they become unreachable, which corresponds to
when they would be popped from the stack in a non-pointer-refined
machine.  Thus the concrete machine really manages continuations like
a stack.

Similarly, in the abstract semantics, continuations are deallocated as
soon as they become unreachable, which often corresponds to when they
would be popped.  We say often, because due to the finiteness of the
store, this correspondence cannot always hold.
However, this approach gives a good finite approximation to infinitary
stack analyses that can always match calls and returns.

\section{Abstract stack inspection}
\label{sec:CM}

In this section, we derive an abstract interpreter for the static
analysis of a higher-order language with stack inspection.  Following
the outline of Section~\ref{sec:cek-to-acesk} and~\ref{sec:krivine},
we start from the tail-recursive CM machine of Clements and
Felleisen~\cite{dvanhorn:Clements2004Tailrecursive}, perform a pointer
refinement on continuations, then abstract the semantics by a
parameterized bounding of the store.

\subsection{The \texorpdfstring{$\boldsymbol{\lambda_{\mathrm{sec}}}$}{lambda-sec}-calculus and stack-inspection}

The $\lambda_{\mathrm{sec}}$-calculus of Pottier, Skalka, and Smith is
a call-by-value $\lambda$-calculus model of higher-order stack
inspection~\cite{dvanhorn:Pottier2005Systematic}.
We present the language as given by Clements and
Felleisen~\cite{dvanhorn:Clements2004Tailrecursive}.

All code is statically annotated with a given set of permissions $R$,
chosen from a fixed set $\Perm$.  A computation whose source code was
statically annotated with a permission may \emph{enable} that
permission for the dynamic extent of a subcomputation.  The
subcomputation is privileged so long as it is annotated with the same
permission, and every intervening procedure call has likewise been
annotated with the privilege.
\[
\begin{grammar}
  \expr &\in& \syn{Exp} &\produces&\dots  
  \opor \failexpr \opor \grantform{R}{\expr}
  \opor
  \\
  &&&& 
  \testform{R}{\expr}{\expr}
  \opor \frameform{R}{\expr}
\end{grammar}
\]
A $\failexpr$ expression signals an exception if evaluated; by
convention it is used to signal a stack-inspection failure.
A $\frameform{R}{\expr}$ evaluates $\expr$ as the principal $R$,
representing the permissions conferred on $\expr$ given its origin.
A $\grantform{R}{\expr}$ expression evaluates as $\expr$ but with the
permissions extended with $R$ enabled.
A $\testform{R}{\expr_0}{\expr_1}$ expression evaluates to $\expr_0$
if $R$ is enabled and $\expr_1$ otherwise.

A trusted annotator consumes a program and the set of permissions it
will operate under and inserts frame expressions around each
$\lambda$-body and intersects all grant expressions with this set of
permissions.  We assume all programs have been properly annotated.

Stack inspection can be understood in terms of an $\OK$ predicate on
an evaluation contexts and permissions.
The predicate determines whether the given permissions are enabled
for a subexpression in the hole of the context.
The $\OK$ predicate holds whenever the context can be traversed from
the hole outwards and, for each permission, find an enabling grant
context without first finding a denying frame context.


\subsection{The CM machine}

The CM (continuation-marks) machine of Clements and Felleisen is a
properly tail-recursive extended CESK machine for interpreting
higher-order languages with
stack-inspection~\cite{dvanhorn:Clements2004Tailrecursive}.

In the CM machine, continuations are annotated with
\emph{marks}~\cite{dvanhorn:Clements2001Modeling}, which, for the
purposes of stack-inspection, are finite maps from permissions
to $\{\no,\grant\}$:
\[
\begin{array}{rcl}
\cont & \produces & \mtk^m \opor \ark^m(\expr,\env,\cont)
\opor \fnk^m(\den,\env,\cont)\text.
\end{array}
\]
We write $\cont[R\mapsto c]$ to mean update the marks on $\cont$ to
$m[R\mapsto c]$.

The CM machine is defined in Figure~\ref{fig:cm}
(transitions that are straightforward adaptations of the
corresponding CESK$^\star$ transitions to incorporate continuation
marks are omitted).
It relies on the
$\OK$ predicate to determine whether the permissions in $R$ are
enabled.
The $\OK$ predicate performs the traversal of the context (represented
as a continuation) using marks to determine which permissions have
been granted or denied.

The semantics of a program is given by the set of reachable states
from an initial machine configuration:
\[
\inj_{\mathit{CM}}(\expr) = \langle\expr,\mte,[\addr_0 \mapsto \mtk^\emptyset],\addr_0\rangle\text.
\]

\begin{figure}
\begin{gather*}
\begin{array}{l|r}
\multicolumn{2}{c}{\state \longmapsto_{\mathit{CM}} \state'} \\[1mm]
\hline  &
\\
\langle\failexpr,\env,\store,\cont\rangle &
\langle\failexpr,\env,\store,\mtk^\emptyset\rangle
\\[1mm]
\langle\frameform{R}{\expr},\env,\store,\cont\rangle &
\langle\expr,\env,\store,\cont[\overline{R}\mapsto\no]\rangle
\\[1mm]
\langle\grantform{R}{\expr},\env,\store,\cont\rangle &
\langle\expr,\env,\store,\cont[R \mapsto\grant]\rangle
\\[1mm]
\langle\testform{R}{\expr_0}{\expr_1},\env,\store,\cont\rangle &
 \begin{cases}
\langle\expr_0,\env,\store,\cont\rangle  & \mbox{ if }\OK(R,\cont),
\\
\langle\expr_1,\env,\store,\cont\rangle &\mbox{ otherwise}
\end{cases}
\end{array}
\\
\begin{array}{@{}rcl@{}}
\OK(\emptyset,\cont)\\[1mm]
\OK(R,\mtk^m) & \iff & (R\cap m^{-1}(\no)=\emptyset)\\[1mm]
\left.
\begin{array}{l}
\OK(R,\fnk^m(\den,\env,\cont)) \\
\OK(R,\ark^m(\expr,\env,\cont)) 
\end{array}
\right\}
& \iff &
\begin{array}{@{}l}
(R \cap m^{-1}(\no) = \emptyset)\;\wedge \\
\ \ \OK(R \setminus m^{-1}(\grant),\cont)
\end{array}
\end{array}
\end{gather*}
\caption{The CM machine and $\OK$ predicate.}
\label{fig:cm}
\end{figure}

\subsection{The abstract CM\texorpdfstring{$^\star$}{*} machine}

Store-allocating continuations, time-stamping, and
bounding the store yields the transition system given in
Figure~\ref{fig:acm}.
The notation $\astore(\addr)[R\mapsto c]$ is used to mean $[R\mapsto
  c]$ should update \emph{some} continuation in $\astore(\addr)$, \ie,
\[
\astore(\addr)[R\mapsto c] = \astore[\addr \mapsto \astore(\addr)
  \setminus \{\cont\}\cup \{\cont[R\mapsto c]\}]\text,
\]
for some $\cont \in \astore(\addr)$.
It is worth noting that continuation marks are updated, not joined, in
the abstract transition system.

\begin{figure}
\begin{gather*}
\begin{array}{l|r}
\multicolumn{2}{c}{\astate \longmapsto_{\widehat{\mathit{CM}^\star}} \astate'} \\[2mm]
\hline & \\
\langle\failexpr,\env,\astore,\addr\rangle &
\langle\failexpr,\env,\astore,\addr_\mtk\rangle
\\[1mm]
\langle\frameform{R}{\expr},\env,\astore,\addr\rangle 
&
\langle\expr,\env,\astore(\addr)[\overline{R}\mapsto\no],\addr\rangle
\\[1mm]
\langle\grantform{R}{\expr},\env,\astore,\addr\rangle 
&
\langle\expr,\env,\astore(\addr)[R \mapsto\grant],\addr\rangle
\\[1mm]
\langle\testform{R}{\expr_0}{\expr_1},\env,\astore,\addr\rangle 
&
\begin{cases}
\langle\expr_0,\env,\astore,\addr\rangle  & \mbox{ if }\AOK(R,\astore,\addr),
\\ 
\langle\expr_1,\env,\astore,\addr\rangle &\mbox{ otherwise.}
\end{cases}
\end{array}
\\
\begin{array}{@{}r@{\;}c@{\;}l}
\AOK(\emptyset,\astore,\addr)\\[1mm]
\AOK(R,\astore,\addr) & \iff & (R\cap m^{-1}(\no)=\emptyset)\\
\mbox{if }\astore(\addr) \ni \mtk^m\\[1mm]
\AOK(R,\astore,\addr) 
& \iff &
(R \cap m^{-1}(\no) = \emptyset)\;\wedge
\\
\mbox{if }\astore(\addr) \ni \fnk^m(\den,\env,\addrnext)
&&
\ \AOK(R \setminus m^{-1}(\grant),\astore,\addrnext)
\\
\mbox{or }\astore(\addr) \ni \ark^m(\expr,\env,\addrnext)
\end{array}
\end{gather*}
\caption{The abstract CM$^\star$ machine.}
\label{fig:acm}
\end{figure}

The $\AOK$ predicate (Figure~\ref{fig:acm}) approximates the pointer
refinement of its concrete counterpart $\OK$, which can be understood
as tracing a path through the store corresponding to traversing the
continuation.  The abstract predicate holds whenever there exists such
a path in the abstract store that would satisfy the concrete
predicate:
Consequently, in analyzing $\testform{R}{\expr_0}{\expr_1}$, $\expr_0$
is reachable only when the analysis can prove the $\OK^\star$
predicate holds on some path through the abstract store.

It is straightforward to define a structural abstraction map and
verify the abstract CM$^\star$ machine is a sound approximation of its
concrete counterpart:
\begin{theorem}[Soundness of the Abstract CM$^\star$ Machine]\ \\
  If $\state \longmapsto_{\mathit{CM}} \state' \text{ and }
  \alpha(\state) \wt \astate$, then there exists an abstract state
  $\astate'$, such that $\astate
  \longmapsto_{\widehat{\mathit{CM}^\star_\tm}} \astate'$ and
  $\alpha(\state') \wt \astate'$.
\end{theorem}

\section{Widening to improve complexity}
\label{sec:widening}

If implemented na\"ively, it takes time exponential in the size of the
input program to compute the reachable states of the abstracted machines.
Consider the size of the state-space for the abstract time-stamped
\CESP{} machine:
\begin{align*}
& \abs{\syn{Exp}
    \times 
    \s{Env}
    \times
    \sa{Store}
    \times 
    \s{Addr}
    \times
    \s{Time}
  }
\\ 
=\; & \abs{\syn{Exp}}
\times
\abs{\s{Addr}}^{\abs{\syn{Var}}}
\times
\abs{\Storable}^{\abs{\s{Addr}}}
\times 
\abs{\s{Addr}}
\times
\abs{\s{Time}}
\text.
\end{align*}
Without simplifying any further, we clearly have an exponential number
of abstract states.

To reduce complexity, we can employ widening in the form of Shivers's
single-threaded store~\cite{mattmight:Shivers:1991:CFA}.
To use a single threaded store, we have to reconsider the abstract evaluation function itself.
Instead of seeing it as a function that returns the set of reachable
states, it is a function that returns a set of partial states plus a
single globally approximating store, \ie, $\avf : \syn{Exp} \to
\s{System}$, where:
\begin{align*}
  \s{System} &= \Pow{\syn{Exp} \times \s{Env} \times \s{Addr} \times \s{Time}} \times \sa{Store}
  \text.
\end{align*}
We compute this as a fixed point of a monotonic function, $f$:
\begin{align*}
  f &: \s{System} \to \s{System}
  \\
  f(C,\astore) &= (C',\astore'') \text{ where }
  \\
  Q' &= \setbuild{ (c',\astore') }{ c \in C \text{ and } (c,\astore) \longmapsto (c',\astore') } 
  \\
  (c_0,\astore_0) &\cong \inj(\expr)
  \\
  C' &= C \union \setbuild{ c' }{ (c',\_) \in Q' } \union \set{c_0}
  \\
  \astore'' &= \astore \join \bigjoin_{  (\_,\astore') \in Q' }  \astore'
\text,
\end{align*}
so that \(\avf(\expr) = \lfp(f)\).
The maximum number of
iterations of the function $f$
times the cost of each iteration
bounds the complexity of the analysis.

\paragraph{Polynomial complexity for monovariance:}
It is straightforward to compute the cost of a monovariant (in our
framework, a ``0CFA-like'') analysis with this widening.
In a monovariant analysis, environments disappear; a monovariant
system-space simplifies to:
\begin{align*}
& \s{System}_0
  \\
  =&\
  \Pow{
    \syn{Exp} \times
    \syn{Lab} \times 
    \syn{Lab}_\bot} 
  \\
  & \times 
  (
  \overbrace{(\syn{Var} + \syn{Lab})}^{\text{addresses}} \to 
  \overbrace{(\syn{Exp} \times \syn{Lab})}^{\text{\textbf{fn} conts}} + 
  \overbrace{(\syn{Exp} \times \syn{Lab})}^{\text{\textbf{ar} conts}} + 
  \syn{Lam}
  )
  \text.
\end{align*}
If ascended monotonically, one could add one new partial state each
time or introduce a new entry into the global store.
Thus, the maximum number of monovariant iterations is:
\begin{align*}
  &\abs{\syn{Exp}} \times
  \abs{\syn{Lab}}^2
  +1
 \\
  +\; &
  \abs{\syn{Var} + \syn{Lab}} 
  \times
  (2\abs{\syn{Exp} \times \syn{Lab}} + \abs{\syn{Lam}})
  \text,
\end{align*}
which is cubic in the size of the program. 

\section{Related work}

The study of abstract machines for the $\lambda$-calculus began with
Landin's SECD machine~\cite{dvanhorn:landin-64}, the theory of
abstract interpretation with the POPL papers of the
Cousots'~\cite{mattmight:Cousot:1977:AI,mattmight:Cousot:1979:Galois},
and static analysis of the $\lambda$-calculus with Jones's coupling of
abstract machines and abstract
interpretation~\cite{mattmight:Jones:1981:LambdaFlow}.
All three have been active areas of research since their inception,
but only recently have well known abstract machines been connected
with abstract interpretation by Midtgaard and
Jensen~\cite{dvanhorn:midtgaard-jensen-sas-08,dvanhorn:Midtgaard2009Controlflow}.
We strengthen the connection by demonstrating a general technique for
abstracting abstract machines.

\paragraph{Abstract interpretation of abstract machines:}

%

The approximation of abstract machine states for the analysis of
higher-order languages goes back to
Jones~\cite{mattmight:Jones:1981:LambdaFlow},
who argued abstractions of regular tree automata could solve the
problem of recursive structure in environments.
We re-invoked that wisdom to eliminate the recursive structure of
continuations by allocating them in the store.

Midtgaard and Jensen present a 0CFA for a CPS $\lambda$-calculus
language~\cite{dvanhorn:midtgaard-jensen-sas-08}.  The approach is
based on Cousot-style calculational abstract
interpretation~\cite{dvanhorn:Cousot98-5}, applied to a functional
language.  Like the present work, Midtgaard and Jensen start with an
``off-the-shelf'' abstract machine for the concrete semantics (in this
case, the CE machine of Flanagan, \etal~\cite{dvanhorn:Flanagan1993Essence}) and employ a reachable-states
model.  They then compose well-known Galois connections to reveal a
0CFA with reachability in the style of
Ayers~\cite{dvanhorn:ayers-phd93}.\footnote{Ayers derived an abstract
  interpreter by transforming (the representation of) a denotational
  continuation semantics of Scheme into a state transition system (an
  abstract machine), which he then approximated using Galois
  connections~\cite{dvanhorn:ayers-phd93}.}  The CE machine is not
sufficient to interpret direct-style programs, so the analysis is
specialized to programs in continuation-passing style.  Later work by
Midtgaard and Jensen went on to present a similar calculational
abstract interpretation treatment of a monomorphic CFA for an ANF
$\lambda$-calculus~\cite{dvanhorn:Midtgaard2009Controlflow}.  The
concrete semantics are based on reachable states of the C$_a$EK
machine~\cite{dvanhorn:Flanagan1993Essence}.  The abstract semantics
approximate the control stack component of the machine by its top
element, which is similar to the labeled machine abstraction given in
Section~\ref{sec:labelled-acesk} when $k=0$.

Although our approach is not calculational like Midtgaard and
Jensen's, it continues in their tradition by applying abstract
interpretation to off-the-shelf tail-recursive machines.  We extend
the application to direct-style machines for a $k$-CFA-like
abstraction that handles tail calls, laziness, state, exceptions,
first-class continuations, and stack inspection.  We have extended
\emph{return flow analysis} to a completely direct style (no ANF or
CPS needed) within a framework that accounts for polyvariance.

Harrison gives an abstract interpretation for a higher-order language
with control and state for the purposes of automatic
parallelization~\cite{mattmight:Harrison:1989:Parallelization}.
Harrison maps Scheme programs into an imperative intermediate
language, which is interpreted on a novel abstract machine.  The
machine uses a procedure string approach similar to that given in
Section~\ref{sec:labelled-acesk} in that the store is addressed by
procedure strings.
%
Harrison's first machine employs higher-order values to represent
functions and continuations and he notes, ``the straightforward
abstraction of this semantics leads to abstract domains containing
higher-order objects (functions) over reflexive domains, whereas our
purpose requires a more concrete compile-time representation of the
values assumed by variables. We therefore modify the semantics such
that its abstraction results in domains which are both finite and
non-reflexive.''
Because of the reflexivity of denotable values, a direct abstraction is not
possible, so he performs closure conversion on the (representation of)
the semantic function.
Harrison then abstracts the machine by bounding the procedure string
space (and hence the store) via an abstraction he calls stack
configurations, which is represented by a finite set of members, each
of which describes an infinite set of procedure strings. 

To prove that Harrison's abstract interpreter is correct he argues
that the machine interpreting the translation of a program in the
intermediate language corresponds to interpreting the program as
written in the standard semantics---in this case, the denotational
semantics of R$^3$RS.  On the other hand, our approach relies on well
known machines with well known relations to calculi, reduction
semantics, and other
machines~\cite{mattmight:Felleisen:1987:Dissertation,dvanhorn:Danvy:DSc}.
These connections, coupled with the strong similarities
between our concrete and abstract machines, result in minimal proof
obligations in comparison.  Moreover, programs are analyzed in
direct-style under our approach.





\paragraph{Abstract interpretation of lazy languages:}

Jones has analyzed non-strict functional
languages~\cite{mattmight:Jones:1981:LambdaFlow,dvanhorn:Jones2007Flow},
but that work has only focused on the by-name aspect of laziness and
does not address memoization as done here.
Sestoft examines flow analysis for lazy languages and uses abstract
machines to prove
soundness~\cite{mattmight:Sestoft:1991:Dissertation}.  In particular,
Sestoft presents a lazy variant of Krivine's machine similar to that
given in Section~\ref{sec:krivine} and proves analysis is sound with
respect to the machine.  Likewise, Sestoft uses Landin's SECD machine
as the operational basis for proving globalization optimizations
correct.  Sestoft's work differs from ours in that analysis is
developed separately from the abstract machines, whereas we derive
abstract interpreters directly from machine definitions.
Fax\'en uses a type-based flow analysis approach to analyzing a
functional language with explicit thunks and evals, which is intended
as the intermediate language for a compiler of a lazy
language~\cite{dvanhorn:Faxen1995Optimizing}.
In contrast, our approach makes no assumptions about the typing
discipline and analyzes source code directly.




\paragraph{Realistic language features and garbage collection:}

Static analyzers typically hemorrhage precision in the presence of
exceptions and first-class continuations:
they jump to the top of the lattice of approximation when these
features are encountered.
%
%
Conversion to continuation- and exception-passing style can handle
these features without forcing a dramatic ascent of the lattice of
approximation~\cite{mattmight:Shivers:1991:CFA}.
The cost of this conversion, however, is lost knowledge---both
approaches obscure static knowledge of stack structure, by desugaring
it into syntax.


Might and Shivers introduced the idea of using abstract garbage
collection to improve precision and efficiency in flow
analysis~\cite{mattmight:Might:2006:GammaCFA}.
They develop a garbage collecting abstract machine for a CPS language
and prove it correct.  We extend abstract garbage collection to
direct-style languages interpreted on the CESK machine.




\paragraph{Static stack inspection:}

Most work on the static verification of stack inspection has focused
on type-based
approaches.
Skalka and Smith present a type system for static enforcement of
stack-inspection~\cite{dvanhorn:Skalka2000Static}.
Pottier \etal~present type systems for enforcing
stack-inspection developed via a static correspondence to the dynamic
notion of security-passing
style~\cite{dvanhorn:Pottier2005Systematic}.
Skalka \etal~present type and effect systems that use
linear temporal logic to express regular properties of program traces
and show how to statically enforce both stack- and history-based
security mechanisms~\cite{dvanhorn:Skalka-Smith-VanHorn:JFP08}.
Our approach, in contrast, is not typed-based and focuses only on
stack-inspection, although it seems plausible the approach of
Section~\ref{sec:CM} extends to the more general history-based
mechanisms.

\section{Conclusions and perspective}

We have demonstrated the utility of store-allocated continuations by
deriving novel abstract interpretations of the CEK, a lazy variant of
Krivine's, and the stack-inspecting CM machines.  These abstract
interpreters are obtained by a straightforward pointer refinement and
structural abstraction that bounds the address space, making the
abstract semantics safe and computable.  Our technique allows concrete
implementation technology to be mapped straightforwardly into that of
static analysis, which we demonstrated by incorporating abstract
garbage collection and optimizations to avoid abstract space leaks,
both of which are based on existing accounts of concrete GC and space
efficiency.  Moreover, the abstract interpreters properly model
tail-calls by virtue of their concrete counterparts being properly
tail-call optimizing.  Finally, our technique uniformly scales up to
richer language features.  We have supported this by extending the
abstract CESK machine to analyze conditionals, first-class control,
exception handling, and state.  We speculate that store-allocating
bindings and continuations is sufficient for a straightforward
abstraction of most existing machines.

\paragraph{Acknowledgments:}  

We thank Matthias Felleisen, Jan Midtgaard, and Sam Tobin-Hochstadt
for discussions and suggestions.
We also thank the anonymous reviewers for their close reading and
helpful critiques; their comments have improved this paper.

\bibliographystyle{plain}

\end{document}